\documentclass[twoside]{IEEEtran} %

\usepackage[utf8]{inputenc} 
\usepackage[T1]{fontenc}
\usepackage{url}
\usepackage{ifthen}
\usepackage{cite}
\usepackage[cmex10]{amsmath} %
\usepackage{enumerate}
\usepackage{array}

\let\truemathcal\mathcal
\usepackage{amssymb,bbm}
\usepackage[mathcal]{euscript}
\usepackage{mathrsfs}

\usepackage{amsthm,proof}
\usepackage{hyperref}

\usepackage{microtype}

\usepackage{footnote}

\usepackage[pdftex]{graphicx}
\DeclareGraphicsExtensions{.pdf,.jpg,.png} %
\usepackage{xcolor}

\newcommand{\E}{\ensuremath\mathbb{E}}
\renewcommand{\P}{\ensuremath\mathbb{P}}
\renewcommand{\d}{\ensuremath\,\mathrm{d}}
\newcommand{\dx}{\ensuremath\,\mathrm{d}x}
\newcommand{\dt}{\ensuremath\,\mathrm{d}t}

\newtheorem{theorem}{Theorem}
\newtheorem{lemma}{Lemma}
\newtheorem{proposition}{Proposition}
\newtheorem{corollary}{Corollary}
\theoremstyle{definition}

\newtheorem{remark}{Remark}
\newtheorem{example}{Example}

\def\qed{\hfill\IEEEQED}

\let\leq\leqslant
\let\geq\geqslant
\let\eps\varepsilon
\let\phi\varphi
\newcommand{\R}{\ensuremath\mathbb{R}}
\newcommand{\Z}{\ensuremath\mathbb{Z}}
\newcommand{\N}{\ensuremath\mathbb{N}}

\begin{document}
\title{Variations on a Theme by Massey} %
\author{%
Olivier Rioul,~\IEEEmembership{Member,~IEEE}\thanks{%
Manuscript received March, 6th, 2021; revised October 30th, 2021; accepted December 23, 2021. Date of publication , 2022; date of current version January 2nd, 2022.
}%
\thanks{The author is with the LTCI, Télécom Paris, Institut Polytechnique de Paris, F-91120, Palaiseau, France (e-mail: olivier.rioul@telecom-paris.fr).}%
}


\maketitle
\begin{abstract}
In 1994, Jim Massey proposed the guessing entropy as a measure of the difficulty that an attacker has to guess a secret used in a cryptographic system, and established a well-known inequality between entropy and guessing entropy. Over 15 years before, in an unpublished work, he also established a well-known inequality for the entropy of an integer-valued random variable of given variance. In this paper, we establish a link between the two works by Massey in the more general framework of the relationship between discrete (absolute) entropy and continuous (differential) entropy.
Two approaches are given in which the discrete entropy (or Rényi entropy) of an integer-valued variable can be upper bounded using the differential (Rényi) entropy of some suitably chosen continuous random variable.
As an application, lower bounds on guessing entropy and guessing moments are derived in terms of entropy or Rényi entropy (without side information) and conditional entropy or Arimoto conditional entropy (when side information is available).
\end{abstract}

\begin{IEEEkeywords}
Arikan's inequality, discrete vs. differential entropies, generalized Gaussian densities, generalized exponential densities, guessing entropy, guessing moments, guessing with side information, Kullback's inequality, Massey's inequality, Poisson summation formula, Rényi entropies, Rényi-Arimoto conditional entropies.
\end{IEEEkeywords}

\section{Introduction}

\IEEEPARstart{I}{n} an unpublished work in the mid-1970s, later published in the late 1980s~\cite{Massey88}, James L.~Massey proved the following bound on the entropy of an integer-valued random variable $X$ with variance~$\sigma^2$: %
\begin{equation}\label{MI}%
H(X)< \tfrac{1}{2}\log\bigl(2\pi e(\sigma^2+\tfrac{1}{12})\bigr).
\end{equation}
This inequality establishes an interesting connection between the entropy of $X$ and that of a Gaussian random variable. After more than a decade, Massey also established an important inequality for the guessing entropy~\cite{Massey94}:
\begin{equation}\label{MIG}%
G(X)\geq 2^{H(X)-2}+1 \text{ when \hbox{$H(X)\geq 2$} bits}, 
\end{equation}
where again an integer-valued random variable (number of guesses) is involved, the guessing entropy $G(X)$ being defined as the minimum average number of guesses.
Perhaps surprisingly, the two Massey inequalities can be seen as part of a common framework which relates discrete (absolute) and continuous (differential) entropies.

The question of making the link between the entropy $H(X)$ of a discrete random variable $X$ and the entropy $h(\mathcal{X})$ of a continuous random variable $\mathcal{X}$ is not new.
The usual setting is to consider a discrete random variable $X$ whose values are regularly spaced $\Delta$ apart, with some probability distribution $p(x)=\P(X\!=\!x)$ having finite entropy.
As $\Delta\to 0$, $X$ may approach in distribution a continuous random variable~$\mathcal{X}$ with density $f$.
How then the discrete (absolute) entropy 
\begin{equation}
H(X)\triangleq\sum_x p(x) \log\! \frac{1}{p(x)}
\end{equation}
is related to the continuous (differential) entropy 
\begin{equation}\label{eq-h}
h(\mathcal{X})\triangleq\int \!f(x) \log\! \frac{1}{f(x)} \dx
\end{equation}
and how can $H(X)$ be evaluated from $h(\mathcal{X})$?
Similarly (or more generally), for any fixed $\alpha>0$, how is the discrete Rényi $\alpha$-entropy
\begin{equation}
H_\alpha(X) \triangleq \frac{1}{1-\alpha} \log \sum_x p(x)^\alpha 
\end{equation}
related to the continuous Rényi $\alpha$-entropy 
\begin{equation}\label{eq-h-alpha}
h_\alpha(\mathcal{X}) \triangleq \frac{1}{1-\alpha} \log \int\! f(x)^\alpha \dx
\end{equation}
and how can $H_\alpha(X)$ be evaluated from $h_\alpha(\mathcal{X})$? The limiting case $\alpha\to 1$ gives $H_1(X)=H(X)$ and $h_1(\mathcal{X})=h(\mathcal{X})$.

For Shannon's entropy, the classical answer to this question dates back to the 1961 textbook by Reza~\cite[\S\,8.3]{Reza61}, and has also been presented in the classical textbooks~\cite[\S\,1.3]{McEliece} and ~\cite[\S\,8.3]{CoverThomas}.  
The approach is to first consider the continuous variable $\mathcal{X}$ having density~$f$, and then \emph{quantize} it to obtain the discrete~$X$ with step size~$\Delta$. It follows that the integral in~\eqref{eq-h} or in~\eqref{eq-h-alpha} can approximated by a Riemann sum. 
Appendix~\ref{app-reza} generalizes the argument to Rényi entropies.
One obtains the well-known approximation 
$H(X) \approx h(\mathcal{X}) - \log\Delta$ for small $\Delta$, and more generally,
\begin{equation}\label{eq-hH}
H_\alpha(X) \approx h_\alpha(\mathcal{X}) - \log\Delta
\end{equation}
for any $\alpha>0$.
Reza's approximation~\eqref{eq-hH}, however appealing as it may be, is not so convenient for evaluating the discrete entropy of $X$ from the continuous one: It requires an arbitrary small $\Delta$ and the resulting values of~$X$ are in fact not necessarily regularly spaced since they correspond to mean values (Eq.~\eqref{mvt} in Appendix~\ref{app-reza}).

Massey's approach, in an unpublished work in the mid-1970s~\cite{Massey88}, is to write density $f$ as a staircase function whose values are the discrete probabilities. 
Compared to Reza's, Massey's approach somehow goes in the opposite direction: Instead of deriving the discrete $X$ from the continuous $\mathcal{X}$ and expressing the continuous entropy in terms of the discrete one, it starts from the discrete random variable $X$ with regularly spaced values, and adds an independent uniformly distributed random perturbation $\mathcal{U}$ to obtain a ``dithered'' continuous random variable $\mathcal{X}=X+\mathcal{U}$. This is explained in~\cite[Exercice 8.7]{CoverThomas},~\cite{CoverThomasSolutions} which also credits an unpublished work by Frans Willems. By doing so, the discrete entropy is expressed in terms of the continuous one. Remarkably, as stated in Theorem~\ref{thm-one} below,~\eqref{eq-hH}  becomes an \emph{exact} equality
\begin{equation}%
H_\alpha(X) = h_\alpha(\mathcal{X}) - \log\Delta
\end{equation}
where $\Delta$ needs not be arbitrarily small. 

This paper presents various Massey-type bounds on the Shannon entropy as well as on the Rényi entropy of an arbitrary positive order $\alpha>0$, of a discrete random variable using a version of \emph{Kullback's inequality} for exponential families applied to $\mathcal{X}$. An alternative bounding technique is to apply Kullback's inequality not to the continuous variable but directly to an integer-valued variable $X$ using the same exponential family density, combined with the \emph{Poisson summation formula} from Fourier analysis.

As an application, Massey's original inequality~\eqref{MI} can be recovered and improved by removing the constant $\frac{1}{12}$ inside the logarithm at the expense of an additional constant which is exponentially small as $\sigma^2$ increases (Equation~\eqref{eq-rioul2} below) :
\begin{equation}
H(X) < \frac{1}{2}\log(2\pi e \sigma^2) + \frac{2\log e}{e^{2\pi^2\sigma^2}-1}.
\end{equation}
In fact, the additional constant can become negative under some mild conditions and the bound $H(X)< \tfrac{1}{2}\log(2\pi e\sigma^2)$---which is classically obtained for \emph{continuous} random variables---holds for many examples of \emph{integer-valued} random variables including ones whose distribution satisfies an entropic central limit theorem. 

The natural generalization of~\eqref{MI} to Rényi entropies is also easily obtained, e.g.,
\begin{equation}
H_{\frac{1}{2}}({X})< \frac{1}{2}\log\Bigl(4\pi^2\bigl(\sigma^2+\frac{1}{12}\bigr)\Bigr) 
\end{equation}
(see~\eqref{ineq-Massey-alpha} below for the general case). This particular inequality can be improved as (Equation~\eqref{ineq-rioul2-1/2} below)
\begin{equation}
H_{\frac{1}{2}}(X) < \log(2\pi \sigma) + \frac{2\log e}{e^{2\pi\sigma}-1}.
\end{equation}

The method is not only applicable  when $X$ has fixed variance but also when $X>0$ has fixed mean~$\mu$ (and more generally with some fixed $\rho$-th order moment). 
It follows that Massey's lower bound~\eqref{MIG} for the guessing entropy can be easily improved as (Equation~\eqref{ineq-rioul2} below):
\begin{equation}
G(X) > \frac{2^{H(X)}}{e} + \frac{1}{2}. 
\end{equation}
valid for any value of $H(X)$. This inequality also holds in the presence of an observed output $Y$ of a side channel using conditional quantities (Equation~\eqref{ineq-rioul2|y} below):
\begin{equation}
G(X|Y) > \frac{2^{H(X|Y)}}{e} + \frac{1}{2}. 
\end{equation}
The improvement over Massey's original inequality~\eqref{MIG} is particularly important for large values of entropy, by the factor ${4}/{e}$. 
It is quite startling to notice that the approach followed by Massey back in the 1970s~\cite{Massey88} can improve the result of his 1994 paper~\cite{Massey94} so much.

The natural generalization to Rényi entropy $H_\alpha(X)$ (without side information) and to Arimoto's conditional entropy $H_\alpha(X|Y)$ (in the presence of some side information $Y$) reads, e.g., 
\begin{align}
G(X|Y) & >\frac{4}{9}2^{H_{2}({X|Y})}+\frac{1}{2}
\\
G(X|Y) & >  \frac{1}{4}2^{H_{\frac{2}{3}}({X|Y})}+\frac{1}{2}
\end{align}
(see~\eqref{ineq-rioul2-alpha|y} below for the general case).
As shown in this paper, such lower bounds depending of $H_\alpha(X|Y)$ cannot hold in general when $\alpha\leq 1/2$, because the support of $X$ may be infinite.
For $X$ with \emph{finite} support of size $M$, Arikan's inequality~\cite{Arikan96}:
\begin{equation}
G(X|Y) \geq  \frac{2^{H_{\frac{1}{2}}\!(X|Y)}}{1\!+\!\ln M} 
\end{equation}
can be recovered and generalized to values $\alpha<1/2$ by the method of this paper, e.g., 
\begin{equation}
G(X|Y) > \frac{2^{2H_{\frac{1}{3}}(X|Y)} }{2(2M+1)} 
\end{equation}
(see~\eqref{gen-arikan<1/2|} below for a general case).
Inequalities relating guessing entropy to (R\'enyi) entropies have become increasingly popular for practical applications because of scalability properties of entropy (see, e.g.,~\cite{ChoudaryPopescu17,TanasescuChoudaryRioulPopescu21}).

The techniques of this paper can also be applied to the guessing $\rho$-th moment $G_\rho(X|Y)$.
While Arikan's inequality
\begin{equation}
G_\rho(X|Y) \geq  \frac{2^{H_{\!\frac{1}{1+\rho}}\!(X|Y)}}{1\!+\!\ln M},
\end{equation}
holds for $X$ with finite support size $M$, lower bounds independent of $M$ and valid for infinite supports can be obtained for any $\alpha>\frac{1}{1+\rho}$, e.g.,
\begin{align}
G_2(X|Y) &> 2 \cdot \frac{2^{2H(X|Y)}}{\pi e}\\ 
G_3(X|Y) &> \frac{9}{2} \cdot \frac{2^{3H_{1/2}(X|Y)}}{\sqrt{3}\,\pi^3}\\
G_4(X|Y) &>  \frac{10000}{59049} \cdot{2^{4H_{2}(X|Y)}}
\end{align}
among many other inequalities of this kind
(see~\eqref{Grho|H} and~\eqref{Grho|Halpha} below for the general case). 

The remainder of this paper is organized as follows. Based on Massey's approach, a general method for establishing Massey-type inequalities for entropies and $\alpha$-entropies is presented in Section~\ref{sec-two}. An alternative ``mixed'' bounding technique using the Poisson summation formula is presented in Section~\ref{sec-alt}. Section~\ref{massey-ineq} applies the method to integer-valued random variables with fixed moment, support length, variance, or mean. Improved inequalities for fixed variance are derived in Section~\ref{sec-three}. Application to guessing is presented in Section~\ref{sec-guess}, where lower bounds are derived for guessing entropy and $\rho$-guessing entropy (guessing moment of order $\rho$).
Section~\ref{sec-conclusion} concludes and suggests perspectives.

\section{General Approach to Massey's Inequalities}\label{sec-two}

\subsection{Massey's Equivalence}

A general approach to Massey-type bounds first consists in identifying discrete entropies to continuous ones as follows.
\begin{theorem}\label{thm-one}
Let $X$ be a discrete random variable whose values are regularly spaced $\Delta$ apart, and define $\mathcal{X}$ by
\begin{equation}\label{perturb}
\mathcal{X}=X+ \mathcal{U}
\end{equation}
 where $\mathcal{U}$ is a continuous random variable independent of $X$, with support of finite length $\leq \Delta$. Then 
\begin{equation}\label{eq-X+Z}
H_\alpha(X) = h_\alpha(\mathcal{X}) -h_\alpha(\mathcal{U}).
\end{equation}
In particular, if $\mathcal{U}$ is uniformly distributed in an interval of length $\Delta$, then $h_\alpha(\mathcal{U})=\log\Delta$ and the exact equality
\begin{equation}\label{eq-hH=}
H_\alpha(X) = h_\alpha(\mathcal{X}) - \log\Delta
\end{equation}
holds for any $\alpha>0$.
\end{theorem}
\begin{proof}
 See Appendix~\ref{app-mass}.
\end{proof}

\begin{remark}
Theorem~\ref{thm-one} shows a peculiar additivity property of entropy:
\begin{equation}
h_\alpha(X+ \mathcal{U})=H_\alpha(X) + h_\alpha(\mathcal{U}),
\end{equation} 
which does not hold in general when $\mathcal{U}$ has support length $>\Delta$.
\end{remark}

\begin{remark}\label{rmk-Z}
The identity~\eqref{eq-hH=} is invariant by \emph{scaling}: if $s>0$, $H_\alpha(sX)=h_\alpha(s\mathcal{X})-\log(s\Delta)$ is the same as~\eqref{eq-hH=} because of the scaling property $h_\alpha(s\mathcal{X})=h_\alpha(\mathcal{X})+\log s$. %
As a result, one can always set $\Delta=1$ and consider an \emph{integer-valued} random variable $X$.
Hereafter whenever $\mathcal{U}$ is taken uniform we shall always make this assumption.
As a result, \eqref{eq-hH=} simply writes 
\begin{equation}\label{eq-hH==}
H_\alpha(X) = h_\alpha(\mathcal{X}) 
\end{equation}
when $\mathcal{U}$ is uniformly distributed in an interval of length $1$.
This is the original remark by Massey~\cite{Massey88} that discrete and continuous entropies coincide in this case.
\end{remark}

\subsection{Inequalities of the Kullback Type}

The next step in the general approach to Massey's inequalities is to bound continuous entropies $h_\alpha(\mathcal{X})$ using appropriate bounding techniques. The case $\alpha=1$ is familiar:

\begin{theorem}[Kullback's Inequality]
Let $\mathcal{X}$ be a continuous random variable with differential entropy $h(\mathcal{X})$ and $T(x)$ be a nonnegative function such that the ``moment'' $\E[T(\mathcal{X})]=m$ is a fixed quantity.
Then
\begin{equation}\label{ineq-general-continuous}
h(\mathcal{X}) \leq m \log e + \log Z
\end{equation}
where $Z%
=\int\! e^{-T(x)} \dx$. Equality holds if and only if $\mathcal{X}$ has density
\begin{equation}\label{eq-phi}
\phi(x) \triangleq \frac{e^{-T(x)} }{Z}.
\end{equation}
\end{theorem}

\begin{proof}
Let $D(f\|\phi)=\int f \log \frac{f}{\phi}$ be the relative entropy (or Kullback-Leibler divergence) between the density $f$ of $\mathcal{X}$ and density $\phi$.  The \emph{information inequality}~\cite[Thm.\,2.6.3]{CoverThomas} states that $D(f\|\phi)\geq 0$ with equality iff (if and only if) $f=\phi$ a.e.
This gives the well known Gibbs inequality
\begin{equation}\label{gibbs}
h(\mathcal{X}) \leq -\E \log\phi(\mathcal{X}) 
\end{equation}
with equality iff $f=\phi$ a.e. Applying Gibbs' inequality to~\eqref{eq-phi} proves the theorem.
\end{proof}

\begin{remark}
Inequality~\eqref{ineq-general-continuous} is well known (see, e.g.,~\cite[\S\,21]{Rioul18}) and can be seen as a version of \emph{Kullback's inequality}~\cite[\S\,4]{Kullback54} (or the \emph{Kullback-Sanov inequality}~\cite[pp.~23--24]{Sanov57},~\cite[Chap.\,3, Thm.~2.1]{Kullback}) for exponential families parameterized by some $\theta\in\R$. It is more general in the sense that one does not use the condition on ``partition function'' $Z=Z(\theta)$ which would be required for equality to hold. Such a condition would read $\frac{\d}{\d\theta}\log Z(\theta)=-m$ in the case of a natural exponential family $\phi(x) = {e^{-\theta {T'}^{\vphantom{2}}(x)} }/{Z(\theta)}$ where $T'$ does not depend on $\theta$.
\end{remark}

The natural generalization to Rényi entropies is as follows.
\begin{theorem}[$\alpha$-Kullback's Inequality]\label{alpha-kullback}
Let $\mathcal{X}$ be a continuous random variable with differential $\alpha$-entropy $h_\alpha(\mathcal{X})$ and $T(x)$ be a nonnegative function such that the ``moment'' $\E[T(\mathcal{X})]=m$ is a fixed quantity.
Then
\begin{equation}\label{ineq-general-continuous-alpha}
h_\alpha(\mathcal{X}) \leq \frac{\alpha}{1-\alpha} \log m + \log Z_\alpha 
\end{equation}
where $Z_\alpha=\int \!{T(x)}^{\frac{\alpha}{\alpha-1}}\dx$. Equality holds iff $\mathcal{X}$ has density
\begin{equation}\label{eq-phi-alpha}
\phi(x)\triangleq \frac{{T(x)}^{\frac{1}{\alpha-1}}}{Z} 
\end{equation}
where $Z=\int \!{T(x)}^{\frac{1}{\alpha-1}}\dx$.
\end{theorem}

\begin{proof}
Let $D_\alpha(f\|\phi)=\frac{1}{\alpha-1}\log\int f^\alpha \phi^{1-\alpha}$ be the Rényi $\alpha$-divergence~\cite{vanErvenHarremoes14} between the density~$f$ of~$\mathcal{X}$ and density $\phi$. We have $D_\alpha(f\|\phi)\geq 0$ with equality iff $f=\phi$~a.e. 
Denoting the ``escort'' densities of exponent $\alpha$ by $f_\alpha=\frac{f^\alpha}{\int\!f^\alpha}$ and $\phi_\alpha=\frac{\phi^\alpha}{\int\!\phi^\alpha}$, the \emph{relative $\alpha$-entropy}~\cite{LapidothPfister16} \footnote{%
Also named Sundaresan's divergence~\cite{Sundaresan07}. For $\alpha=2$, $D_2(f\|\phi)=\log \frac{\int f^2 \int g^2}{(\int fg)^2}$\smallskip{} was previously known as the Cauchy-Schwarz divergence~\cite[Eq.\,(31) p.\,38]{Principe00}.
}
between $f$ and $\phi$ is defined as
\begin{equation}
\Delta_\alpha (f\|\phi) \triangleq D_{1/\alpha} (f_\alpha\|\phi_\alpha)
\end{equation}
which is nonnegative and vanishes iff $f=\phi$ a.e.
Expanding $D_{1/\alpha} (f_\alpha\|\phi_\alpha)$ gives the $\alpha$-Gibbs' inequality~\cite[Prop.~8]{Rioul20}
which generalizes Gibbs' inequality~\eqref{gibbs}:
\begin{equation}\label{gibbs-alpha}
h_\alpha(\mathcal{X}) \leq \frac{\alpha}{1-\alpha} \log \E \,\phi_\alpha^{1-\frac{1}{\alpha}}(\mathcal{X})
\end{equation}
with equality iff $f=\phi$ a.e. 
Applying $\alpha$-Gibbs' inequality to~\eqref{eq-phi-alpha} proves the theorem.
\end{proof}

\begin{remark}
Notice that both ${T(x)}^{\frac{1}{\alpha-1}}$ and ${T(x)}^{\frac{\alpha}{\alpha-1}}$ need to be Lebesgue-integrable over the given support interval for~$Z$ and $Z_\alpha$ to be well defined and 
finite. %

If the relation $\E[T(\mathcal{X})]=m$ is also satisfied when $\mathcal{X}\sim\phi$, then
\begin{equation}\label{eq-trick}
\frac{Z_\alpha}{Z}=\frac{1}{Z}\int \!T(x){T(x)}^{\frac{1}{\alpha-1}}\dx=\E[T(\mathcal{X})]=m
\end{equation}
so that in this case~\eqref{ineq-general-continuous-alpha} simplifies to
\begin{equation}\label{ineq-general-continuous-alpha-Z}
h_\alpha(\mathcal{X}) \leq \frac{ \log m}{1-\alpha} + \log Z.
\end{equation} 
\end{remark}

\subsection{Examples of Inequalities of the Kullback Type}\label{sec-two-ex}

A general maximization statement of $\alpha$-entropies subject to constraints is given in~\cite{BunteLapidoth16}.
A fairly general example is obtained when $\mathcal{X}$ is parametrized by $\rho\mskip0.5\thinmuskip$th-order moment $\theta=\E(|\mathcal{X}|^\rho)$ where $\rho>0$ is arbitrary. 

\begin{theorem}\label{lem-three}
For  $\theta=\E(|\mathcal{X}|^\rho)$ with $0<\rho<+\infty$, and $\alpha>\frac{1}{1+\rho}$, both~\eqref{ineq-general-continuous-alpha} and~\eqref{ineq-general-continuous-alpha-Z} reduce to
\begin{equation}\label{eq-maxentgen-alpha}
h_\alpha(\mathcal{X})\leq 
\begin{cases}
\frac{1}{\rho} \log \bigl(\frac{(1+\rho)\alpha-1}{1-\alpha}\theta\bigr)
+\frac{1}{1-\alpha}\log \frac{\rho\alpha}{(1+\rho)\alpha-1} 
\\[1ex]\;+ \log \frac{\textcolor{blue}{\bf 2}\cdot\Gamma(\frac{1}{\rho}+1)\Gamma(\frac{1}{1-\alpha}-\frac{1}{\rho})}{\Gamma(\frac{1}{1-\alpha})}
\qquad\;\text{for $\frac{1}{1+\rho}\!<\!\alpha\!<\!1$;}\\[3ex] 
\frac{1}{\rho} \log \bigl(\frac{(1+\rho)\alpha-1}{\alpha-1}\theta\bigr)
+\frac{1}{\alpha-1}\log \frac{(1+\rho)\alpha-1}{\rho\alpha}
\\[1ex]\;+\log\frac{\textcolor{blue}{\bf 2}\cdot \Gamma(\frac{1}{\rho}+1)\Gamma(\frac{\alpha}{\alpha-1})}{\Gamma(\frac{\alpha}{\alpha-1}+\frac{1}{\rho})}
\qquad\quad\;\text{for $\alpha>1$,}
\end{cases}
\end{equation}
with equality iff $\mathcal{X}$ is a generalized $\alpha$-Gaussian random variable. 
Inequality~\eqref{ineq-general-continuous} reduces to
\begin{equation}\label{ineq-gen-rho}
h(\mathcal{X})\leq \tfrac{1}{\rho}\log(\rho e \theta)  + \log \bigl(\textcolor{blue}{\bf 2}\Gamma(1+\tfrac{1}{\rho})\bigr)
\end{equation}
with equality iff $\mathcal{X}$ is a generalized Gaussian random variable.

In case of the one-side constraint $\mathcal{X}\geq 0$ with $\theta=\E({\mathcal{X}}^\rho)$, the same inequalities hold when the factor $\textcolor{blue}{\bf 2}$ inside the logarithm is removed.
\end{theorem}
\begin{proof}
See Appendix~\ref{app-lem}, where the generalized $\alpha$-Gaussian is given in~\eqref{alpha-Gaussian-generalized}. 
The limiting case $\alpha\to 1$ gives~\eqref{ineq-gen-rho}.
The case $\alpha=1$ is also proved directly by setting $T(x)=\frac{1}{\rho}\frac{|x|^\rho}{\theta}$ so that $m=\frac{1}{\rho}$ and $Z=2\Gamma(1+\frac{1}{\rho})(\rho\theta)^{1/\rho}$ in~\eqref{ineq-general-continuous}.
\end{proof}

Let $\mu_{\mathcal{X}}$ and $\sigma^2_{\mathcal{X}}$ denote the mean and variance of ${\mathcal{X}}$, respectively.
We illustrate~Theorem~\ref{lem-three} in three classical situations:
 
\paragraph{Support length parameter} 
This can be seen as a particular case of Theorem~\ref{lem-three} by setting $\rho=+\infty$ in the case of a finite support $(-1,1)$. More generally, suppose $\mathcal{X}$ has finite support:  $\mathcal{X}\in(a,b)$ a.s.; letting $\ell(\cdot)$ denote the support length, the corresponding parameter is $\theta=\ell(\mathcal{X})=b-a$.
For $\alpha=1$, we set $T(x)=0$ if $x\in(a,b)$ and $=+\infty$ otherwise. Then $\phi$ is the uniform distribution on $(a,b)$, moment $m=0$, partition $Z=b-a$ and~\eqref{ineq-general-continuous} reduces to the known bound~\cite[Ex.\,12.2.4]{CoverThomas} 
\begin{equation}\label{maxentunif}
h(\mathcal{X})\leq \log(b-a) 
\end{equation}
with equality iff $\mathcal{X}$ is uniformly distributed in $(a,b)$.

For $\alpha\ne 1$, we set $T(x)=1$ if $x\in(a,b)$ and $=0$ otherwise, so that $\phi=\phi_\alpha$ is the uniform distribution on $(a,b)$, moment $m=1$, $Z=Z_\alpha=b-a$ and~\eqref{ineq-general-continuous-alpha} or~\eqref{ineq-general-continuous-alpha-Z} reduces to
\begin{equation}\label{maxentunif-alpha}
h_\alpha(\mathcal{X})\leq \log(b-a) 
\end{equation}
with equality iff $\mathcal{X}$ is uniformly distributed in $(a,b)$.

\paragraph{Variance parameter}  
This can be seen as a particular case of Theorem~\ref{lem-three} by setting
$\rho=2$ for the centered variable $\mathcal{X}-\mu_{\mathcal{X}}$.
A direct derivation is as follows.
We assume that $\mathcal{X}\!\in\!\R$ with parameter $\theta=\sigma_{\mathcal{X}}$. For $\alpha=1$ we set $T(x)=\frac{1}{2}(\!\frac{x-\mu_{\mathcal{X}}}{\sigma_{\mathcal{X}}}\!)^2$, so that $\phi=\truemathcal{N}(\mu_{\mathcal{X}},\sigma_{\mathcal{X}}^2)$ is the Gaussian density,  moment $m=\frac{1}{2}$, partition $Z=\sqrt{2\pi\sigma^2_{\mathcal{X}}}$, and~\eqref{ineq-general-continuous} reduces to the well-known Shannon bound~\cite[\S\,20.5]{Shannon48} 
\begin{equation}\label{maxentvar}
h(\mathcal{X})\leq \frac{1}{2}\log(2\pi e\sigma_{\mathcal{X}}^2)
\end{equation} 
with equality iff $\mathcal{X}$ is Gaussian. 

For $\alpha\ne 1$ we set $T(x)$ in the form $T(x)=1+\beta\cdot (\frac{x-\mu_{\mathcal{X}}}{\sigma_{\mathcal{X}}})^2$ so that $m=1+\beta$ and $\beta$ is such that~\eqref{eq-phi-alpha} has finite variance~$\sigma_{\mathcal{X}}^2$. 
The corresponding density $\phi$ is known as the \emph{$\alpha$-Gaussian} density~\cite{CostaHeroVignat03}. 
Under these assumptions, one has $\alpha>\frac{1}{3}$, $\beta=\frac{1-\alpha}{3\alpha-1}$, and both~\eqref{ineq-general-continuous-alpha}\,and\,\eqref{ineq-general-continuous-alpha-Z} reduce to the following
\begin{corollary}\label{lem-one}
For any continuous random variable $\mathcal{X}$ with differential $\alpha$-entropy $h_\alpha(\mathcal{X})$,
\begin{equation}\label{maxentvar-alpha}
h_\alpha(\mathcal{X})\leq 
\begin{cases}
\frac{1}{2} \log \bigl(\frac{3\alpha-1}{1-\alpha}\pi\sigma^2_{\mathcal{X}}\bigr)
+\frac{1}{1-\alpha}\log \frac{2\alpha}{3\alpha-1} 
\\[1ex]\;+ \log \frac{\Gamma(\frac{1}{1-\alpha}-\frac{1}{2})}{\Gamma(\frac{1}{1-\alpha})}
\qquad\qquad\qquad\text{for $\frac{1}{3}<\alpha<1$;}\\[3ex]
\frac{1}{2} \log \bigl(\frac{3\alpha-1}{\alpha-1}\pi\sigma^2_{\mathcal{X}}\bigr)
+\frac{1}{\alpha-1}\log \frac{3\alpha-1}{2\alpha}
\\[1ex]\;+\log\frac{\Gamma(\frac{\alpha}{\alpha-1})}{\Gamma(\frac{\alpha}{\alpha-1}+\frac{1}{2})}
\qquad\qquad\qquad\text{for $\alpha>1$,}
\end{cases}
\end{equation}
with equality iff $\mathcal{X}$ is $\alpha$-Gaussian.
\end{corollary}

\begin{proof}
See Appendix~\ref{app-lem}, where the expression of the $\alpha$-Gaussian is given in~\eqref{alpha-Gaussian}.
\end{proof}
Fig.~\ref{fig1} plots $\alpha$-Gaussian densities for different values of~$\alpha$.

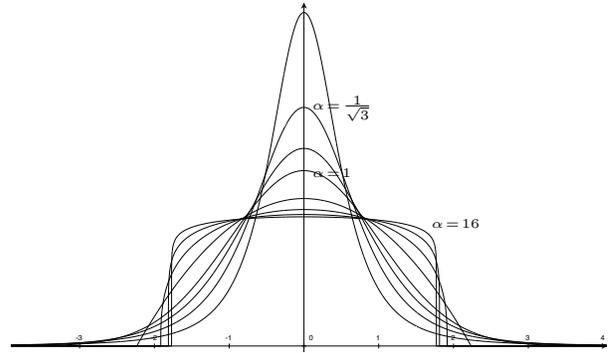
\begin{figure}[htb!]
\centering
\includegraphics[width=.9\linewidth]{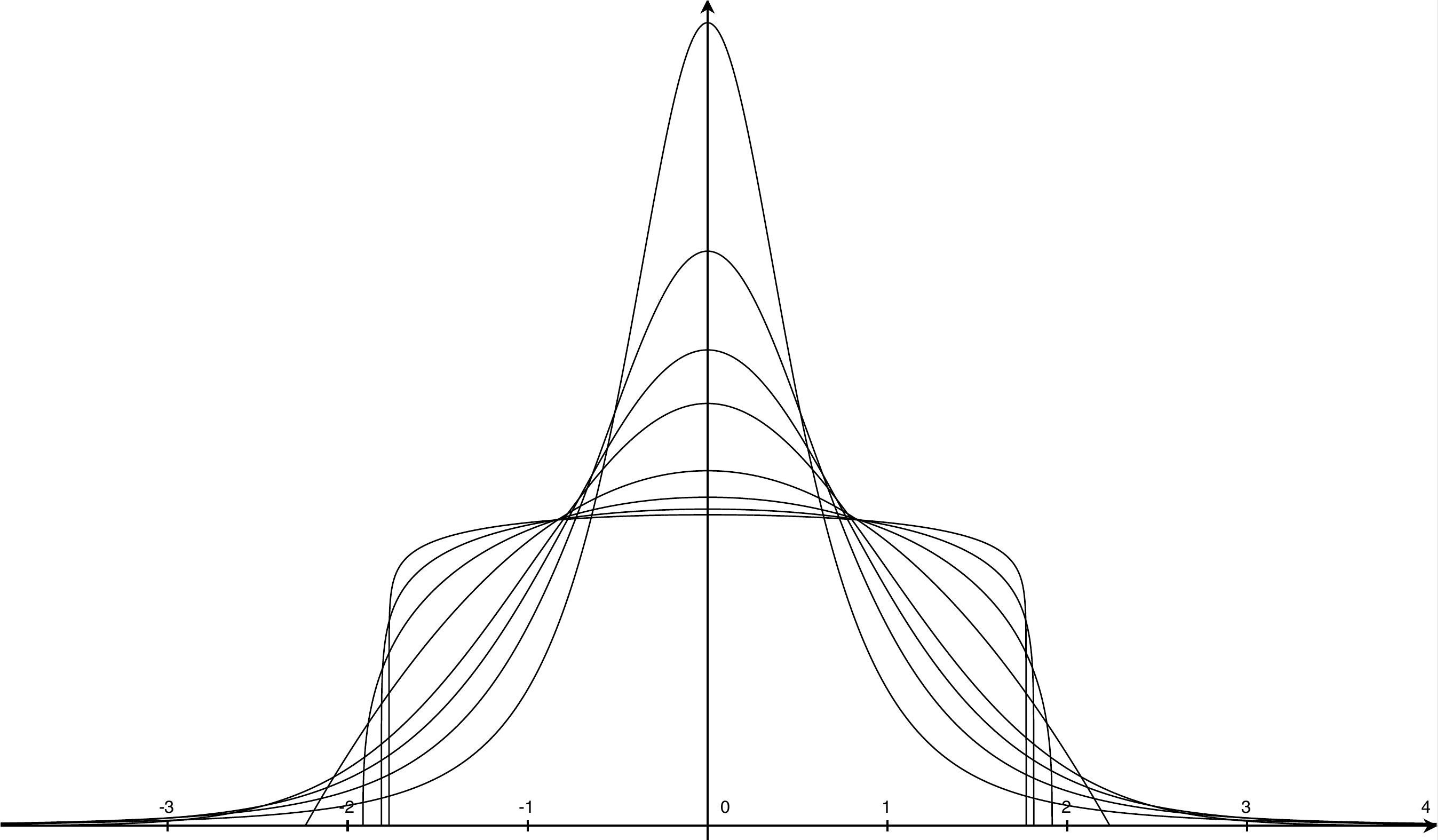}
\setlength{\unitlength}{0.55pt}
\begin{picture}(0,0)
\put(-217,167){\tiny$\alpha\!=\!\frac{1}{\sqrt{3}}$}
\put(-217,120){\tiny$\alpha\!=\!1$}
\put(-135,85){\tiny$\alpha\!=\!16$}
\end{picture}
\caption{$\alpha$-Gaussian distributions~(\protect\ref{alpha-Gaussian}) for $\alpha=3^{-3/4}$, $3^{-1/2}$, $3^{-1/4}$, $1$, $2$, $4$, $8$, $16$.} \label{fig1}
\end{figure}

\begin{example}\label{ex-alpha-sigma}
When $\alpha\to 1$ we recover~\eqref{maxentvar} attained for the Gaussian density. As other examples we have%
\begin{align}
h_{\frac{1}{2}}(\mathcal{X})&\leq \log(2\pi\sigma_{\mathcal{X}})\\
h_{\frac{2}{3}}(\mathcal{X})&\leq  \log\bigl(\frac{8\pi \,\sigma_{\mathcal{X}}}{3\sqrt{3}} \bigr) \\
h_{2}(\mathcal{X})&\leq 
\log\bigl(\frac{5\sqrt{5} \,\sigma_{\mathcal{X}}}{3}\bigr)\\
h_{3}(\mathcal{X})&\leq 
\log\Bigl(\frac{2{\pi\, \sigma_{\mathcal{X}}}}{\sqrt{3}}\Bigr)
\end{align}
with equality iff $\mathcal{X}$ is $\frac{1}{2}$-Gaussian, $\frac{2}{3}$-Gaussian, $2$-Gaussian and $3$-Gaussian, respectively. 
\end{example}

\paragraph{Mean parameter} 
This can be seen as a particular case of Theorem~\ref{lem-three} by setting
$\rho=1$  under the one-sided constraint $\mathcal{X}\geq 0$.
A direct derivation is as follows.
We assume that $\mathcal{X}> 0$ a.s. with parameter $\theta=\mu_{\mathcal{X}}$. For $\alpha=1$ we set $T(x)=\frac{x}{\mu_{\mathcal{X}}}$ so that $\phi$ is the exponential density, moment $m=1$, partition $Z=\mu_{\mathcal{X}}$ and~\eqref{ineq-general-continuous} reduces to another Shannon bound~\cite[\S~20.7]{Shannon48} 
\begin{equation}\label{maxentmean}
h(\mathcal{X})\leq \log(e\mu_\mathcal{X}) 
\end{equation}
with equality iff $\mathcal{X}$ is exponential.

For $\alpha\ne 1$, we set $T(x)$ in the form $T(x)=1+\beta\cdot\frac{x}{\mu_{\mathcal{X}}}$ so that $m=1+\beta$ and $\beta$ is such that~\eqref{eq-phi-alpha} has finite mean $\mu_{\mathcal{X}}$. The corresponding density~$\phi$ can be named ``$\alpha$-exponential''.  
Under these assumptions, one has $\alpha>\frac{1}{2}$, $\beta=\frac{1-\alpha}{2\alpha-1}$, 
and both~\eqref{ineq-general-continuous-alpha}\,and\,\eqref{ineq-general-continuous-alpha-Z} reduce to the following
\begin{corollary}\label{lem-two}
For any continuous random variable $\mathcal{X}$ with differential $\alpha$-entropy $h_\alpha(\mathcal{X})$,
\begin{equation}\label{maxentmean-alpha}
h_\alpha(\mathcal{X}) \!\leq\! \log\mu_{\mathcal{X}} + \frac{\alpha}{1\!-\!\alpha} \log \frac{\alpha}{2\alpha\!-\!1} 
\!=\!\log\mu_{\mathcal{X}} + \frac{\alpha}{\alpha\!-\!1} \log \frac{2\alpha\!-\!1}{\alpha\mathstrut}
\end{equation}
with equality iff $\mathcal{X}$ is $\alpha$-exponential.
\end{corollary}
\begin{proof}
See Appendix~\ref{app-lem}, where the expression of the $\alpha$-exponential is given in~\eqref{alpha-exp}.
\end{proof}
Fig.~\ref{fig2} plots $\alpha$-exponential densities for different values of~$\alpha$. For $\alpha<1$, $\phi$ is a Pareto Type II distribution with shape parameter $\frac{\alpha}{1-\alpha}$, also known as the Lomax density. %

\begin{figure}[htb!]
\centering
\includegraphics[width=.9\linewidth]{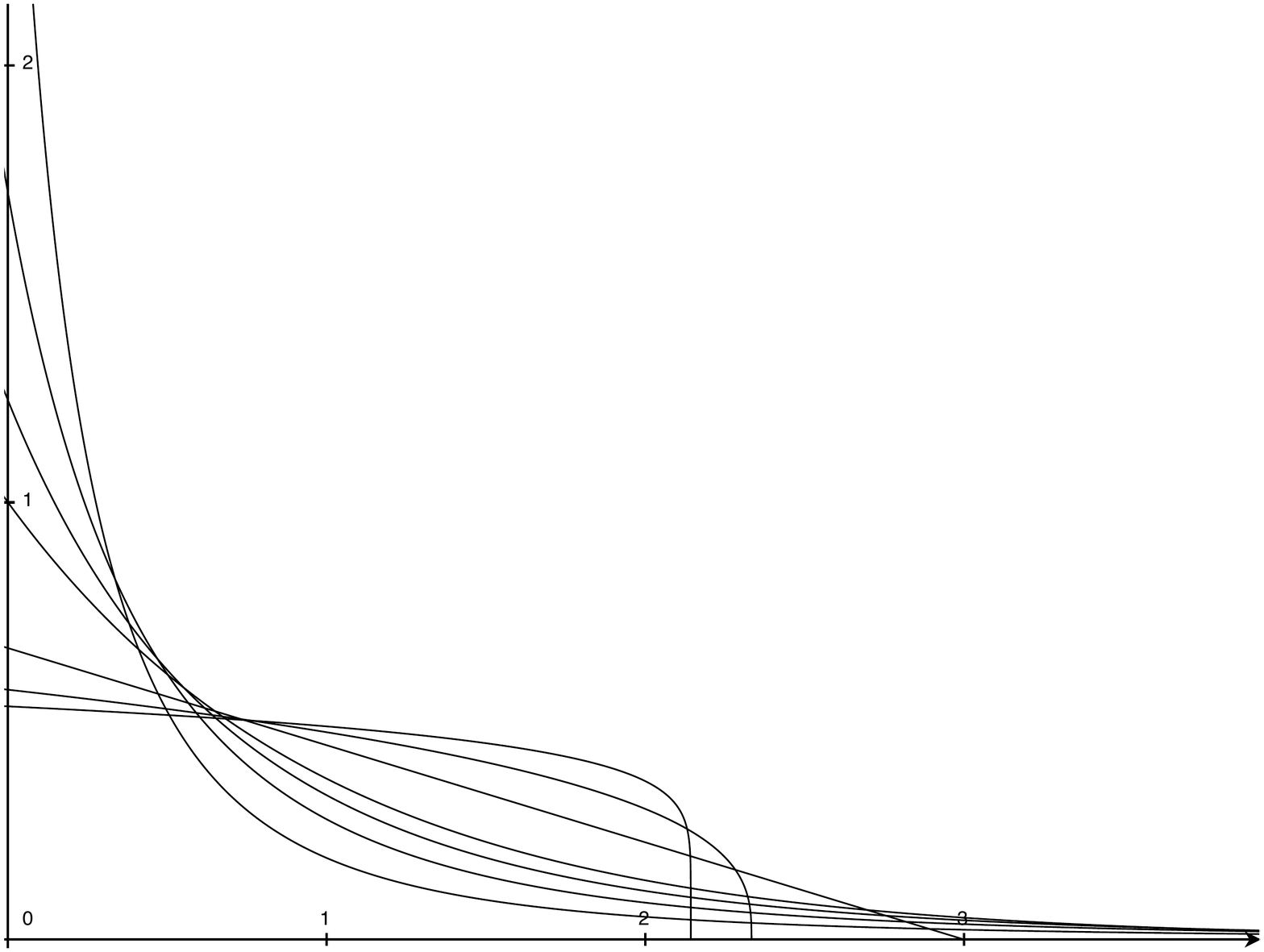}
\setlength{\unitlength}{0.58pt}
\begin{picture}(0,0)
\put(-300,35){\tiny$\alpha\!=\!\tfrac{1}{\sqrt{2}}$}
\put(-255,38){\tiny$\alpha\!=\!1$}
\put(-235,65){\tiny$\alpha\!=\!8$}
\end{picture}
\caption{$\alpha$-exponential distributions~(\protect\ref{alpha-exp}) for $\alpha=2^{-3/4}$, $2^{-1/2}$, $2^{-1/4}$, $1$, $2$, $4$, $8$.} \label{fig2}
\end{figure}

\begin{example}\label{ex-alpha-mu}
When $\alpha\to 1$ we recover~\eqref{maxentmean} attained for the exponential density. As other examples we have%
\begin{align}
h_{\frac{2}{3}}(\mathcal{X})&\leq \log (4\mu_{\mathcal{X}})\\
h_{\frac{3}{4}}(\mathcal{X})&\leq \log\frac{27\mu_{\mathcal{X}}}{8}\\
h_{2}(\mathcal{X})&\leq \log\frac{9\mu_{\mathcal{X}}}{4}
\end{align}
with equality iff $\mathcal{X}$ is $\frac{2}{3}$-exponential, $\frac{3}{4}$-exponential, and $2$-exponential, respectively. 
\end{example}

\section{Alternative Bounding Techniques}\label{sec-alt}

\subsection{Mixed Discrete-Continuous Inequalities of the Kullback Type}

Instead of applying Kullback inequalities~\eqref{ineq-general-continuous} or~\eqref{ineq-general-continuous-alpha} on $\mathcal{X}=X+\mathcal{U}$, it is possible, as an alternative, to apply a similar inequality directly on the discrete entropy of $X$ but using the same  probability density functions.

\begin{theorem}[Case $\alpha=1$]%
\label{thm-two}%
 Let $X$ be a discrete random variable and let $\mathcal{X}$ be the random variable having density~\eqref{eq-phi}:
 \begin{equation}\label{eq-phif}
f(x) \triangleq \frac{e^{-T(x)} }{Z} 
\end{equation}
such that the ``moment'' $\E[T(X)]=\E [T(\mathcal{X})]= m$ is a fixed quantity. Then 
\begin{equation}\label{ineq-general-discrete}
H(X) \leq h(\mathcal{X}) + \log Z'  %
\end{equation}
where %
\begin{equation}\label{eq-Z'}
Z'=\sum_{x} f(x),
\end{equation}
the sum being over all discrete values $x$ of $X$.
\end{theorem}

\begin{proof}
Apply the information inequality $D(p\|q)\geq 0$ to $p(x)=\P(X\!=\!x)$, the probability distribution of $X$,
and to $q(x)=\frac{f(x)}{Z'}$, which is also a discrete probability distribution on the same alphabet because of the normalization constant $Z'$. We obtain Gibbs' inequality in the form $H(X) \leq -\E \log q(X) = -\E \log f(X) +\log Z'$ where $-\E \log f(X) = \E[T(X)] \log e +\log Z =  \E[T(\mathcal{X})] \log e +\log Z=h(\mathcal{X})$ by the equality case in~\eqref{ineq-general-continuous}.
\end{proof}

\begin{theorem}[Case $\alpha\ne 1$]\label{thm-two-alpha}%
 Let $X$ be a discrete random variable and let $\mathcal{X}$ be the random variable having density~\eqref{eq-phi-alpha}: 
\begin{equation}\label{eq-phi-alphaf}
f(x)\triangleq \frac{{T(x)}^{\frac{1}{\alpha-1}}}{Z} 
\end{equation}
such that the ``moment'' $\E[T(X)]=\E [T(\mathcal{X})]= m$ is a fixed quantity. Then 
\begin{equation}\label{ineq-general-discrete-alpha}
H_\alpha(X) \leq h_\alpha(\mathcal{X}) + \log Z'_\alpha  %
\end{equation}
where %
\begin{equation}\label{eq-Z'alpha}
Z'_\alpha=\sum_{x} f_\alpha(x),
\end{equation}
and $f_\alpha=\frac{f^\alpha}{\int f^\alpha}$ is the $\alpha$-escort density of $f$, the sum being over all discrete values $x$ of $X$.
\end{theorem}

\begin{proof}
Let $D_\alpha(p\|q)=\frac{1}{\alpha-1}\log\sum p^\alpha(x) q^{1-\alpha}(x)$ be the Rényi $\alpha$-divergence~\cite{vanErvenHarremoes14} between the distribution~$p$ of a discrete random variable $X$ and some probability distribution $q$ defined over the same alphabet. We have $D_\alpha(p\|q)\geq 0$ with equality iff $p=q$~a.e. 
Denoting the ``escort'' distributions of exponent $\alpha$ by $p_\alpha(x)=\frac{p^\alpha(x)}{\sum p^\alpha(x)}$ and $q_\alpha(x)=\frac{q^\alpha(x)}{\sum q^\alpha(x)}$, the \emph{relative $\alpha$-entropy}~\cite{LapidothPfister16} 
between $p$ and $q$ is defined as
\begin{equation}
\Delta_\alpha (p\|q) \triangleq D_{1/\alpha} (p_\alpha\|q_\alpha) \geq 0
\end{equation}
with equality $=0$ iff $p=q$ a.e.
Expanding $D_{1/\alpha} (p_\alpha\|q_\alpha)$ similarly as in~\cite[Prop.~8]{Rioul20} gives the following $\alpha$-Gibbs' inequality
which generalizes the discrete Gibbs inequality:
\begin{equation}\label{gibbs-alpha-z}
H_\alpha({X}) \leq \frac{\alpha}{1-\alpha} \log \E \,q_\alpha^{1-\frac{1}{\alpha}}({X})
\end{equation}
with equality iff $p=q$ a.e. 
Now apply~\eqref{gibbs-alpha-z} to $p(x)=\P(X\!=\!x)$, the probability distribution of $X$, and to $q(x)=\frac{f(x)}{Z'}$ with the normalization constant $Z'=\sum_x f(x)$, which is also a discrete probability distribution on the same alphabet. Since $q_\alpha(x)=\frac{f_\alpha(x)}{Z'_\alpha}$, we obtain $H_\alpha(X) \leq \frac{\alpha}{1-\alpha} \log \E \,q_\alpha^{1-\frac{1}{\alpha}}({X}) = \frac{\alpha}{1-\alpha} \log \E \,f_\alpha^{1-\frac{1}{\alpha}}({X}) +\log Z'_\alpha$ where \smallskip$\frac{\alpha}{1-\alpha} \log \E \,f_\alpha^{1-\frac{1}{\alpha}}({X})=\frac{\alpha}{1-\alpha}\log \E[T(X)] +\log Z_\alpha =\smallskip  \frac{\alpha}{1-\alpha}\log \E[T(\mathcal{X})] +\log Z_\alpha = h_\alpha(\mathcal{X})$ by the equality case in~\eqref{ineq-general-continuous-alpha}.
\end{proof}

\begin{remark}
Similary as for~\eqref{eq-hH=}, notice that~\eqref{ineq-general-discrete} and~\eqref{ineq-general-discrete-alpha} are invariant by \emph{scaling}: if $\Delta>0$, $H_\alpha(\Delta X)=H_\alpha(X)$ while $h_\alpha(\Delta\mathcal{X})=h_\alpha(\mathcal{X})+\log\Delta$, hence under scaling by $\Delta$, $Z'_\alpha$ is divided by~$\Delta$, and the r.h.s. of~\eqref{ineq-general-discrete-alpha} becomes $h_\alpha(\mathcal{X})+\log\Delta+\log(Z'_\alpha/\Delta)=h_\alpha(\mathcal{X})+\log Z'_\alpha$.
\end{remark}

\subsection{Examples of Mixed Inequalities of the Kullback Type}\label{subsec-mixed-ex}

As in the preceding section, we illustrate the bounding method for an integer-valued $X$ in three situations:
\paragraph{Support length parameter} 
$X$ has finite support $\{k,\ldots,k+\ell\}$ of length $\ell\geq 0$, $\mathcal{X}$ is uniformly distributed on an interval $(a,b)$ that includes $\{k,k+\ell\}$. Then $h(\mathcal{X})=h_\alpha(\mathcal{X})=\log (b-a)$,
$f=f_\alpha$, $Z'_\alpha=\sum_x \frac{1}{b-a} = \frac{\ell+1}{b-a}$ so that~\eqref{ineq-general-discrete} and~\eqref{ineq-general-discrete-alpha} reduce to the known bound $H_\alpha(X)\leq \log (b-a) + \log \frac{\ell+1}{b-a} = \log(\ell+1)$ achieved when $X$ is equiprobable.

\paragraph{Variance parameter} 

\begin{corollary}
Let $X$ be integer-valued with finite mean $\mu$ and variance $\sigma^2$.
Then
\begin{equation}\label{ineq-sigma}
H(X) \leq  \tfrac{1}{2}\log(2\pi e\sigma^2) +\log \sum_x \frac{e^{-\frac{1}{2}(\frac{x-\mu}{\sigma})^2}}{\sqrt{2\pi\sigma^2}},
\end{equation} 
which can be simplified as
\begin{equation}\label{ineq-sigmabis}
H(X) \leq  \frac{\log e}{2} +\log \sum_x e^{-\frac{1}{2}(\frac{x-\mu}{\sigma})^2},
\end{equation}
the sums being taken over all nonnegative integer values $x$ of~$X$.

For $\alpha>\frac{1}{3}$ and any integer-valued $X$ with mean $\mu$ and variance $\sigma^2$, 
\begin{equation}\label{ineq-sigma-alpha}
H_\alpha(X)\leq   \frac{\alpha}{1\!-\!\alpha}\log \frac{2\alpha}{3\alpha\!-\!1} 
+\log\sum_x
\Bigl(1+\frac{1\!-\!\alpha}{3\alpha\!-\!1} \bigl(\frac{x\!-\!\mu}{\sigma}\bigr)^2\Bigr)_{\!\!+}^{\frac{\alpha}{\alpha\!-\!1}}
\end{equation}
where the sum is taken over all integer values $x$ of $X$.
\end{corollary}

\begin{proof}
For $\alpha=1$ we take $\mathcal{X}\sim\mathcal{N}(\mu,\sigma^2)$ of differential entropy $h(\mathcal{X})=\frac{1}{2}\log(2\pi e\sigma^2)$.
Theorem~\ref{thm-two} then gives~\eqref{ineq-sigma}.

For $\alpha\ne 1$ we take $\mathcal{X}$ to be $\alpha$-Gaussian of parameters $(\mu_{\mathcal{X}}=\mu,\sigma^2_{\mathcal{X}}=\sigma^2)$ and differential entropy $h_\alpha(\mathcal{X}) = \frac{\alpha}{1-\alpha} \log (1+\beta) + \log Z_\alpha$, given by the r.h.s. of~\eqref{maxentvar-alpha}. 
From the expression of an $\alpha$-Gaussian~\eqref{alpha-Gaussian}, we have
$
f_\alpha(x) = \frac{1}{Z_\alpha} \bigl(1+\beta (\frac{x-\mu}{\sigma})^2\bigr)_+^{\frac{\alpha}{\alpha-1}}
$
where $\beta=\frac{1-\alpha}{3\alpha-1}$ and $Z_\alpha$ is given by~\eqref{eq-Zalpha-sigma}.
Therefore, Theorem~\ref{thm-two-alpha} gives~\eqref{ineq-sigma-alpha}.
\end{proof}

\begin{remark}\label{rmk-muindep}
It may appear peculiar that the upper bound in~\eqref{ineq-sigma},~\eqref{ineq-sigmabis}  or~\eqref{ineq-sigma-alpha} depends on the mean $\mu=\E(X)$ while the entropy $H_\alpha(X)$ should not. But this upper bound is, in fact, invariant by translation $X+c$ (where $c\in\Z$ because of the constraint of integer-valued variables), as is readily seen by making a change of variables in the sum, e.g., $\sum_x e^{-\frac{1}{2}(\frac{x-(\mu+c)}{\sigma})^2}=\sum_x e^{-\frac{1}{2}(\frac{x-\mu}{\sigma})^2}$. In other words, the upper bound in~\eqref{ineq-sigma},~\eqref{ineq-sigmabis} or~\eqref{ineq-sigma-alpha} depends only on $\mu$'s fractional part $\{\mu\}=\mu\bmod1$.
\end{remark}

\begin{remark}\label{rmk-partialsum}
The sum in~\eqref{ineq-sigma},~\eqref{ineq-sigmabis} or~\eqref{ineq-sigma-alpha} does not need to be taken over \textit{all} integers if the support of $X$ is limited. A tighter bound always results if one takes the sum only on those integers actually taken by the variable.
In particular, when $\alpha>1$, the sum in~\eqref{ineq-sigma-alpha} is restricted to values $x$ in the interval $|x-\mu|<\sqrt{\frac{3\alpha-1}{\alpha-1}}$.  
\end{remark}

\begin{remark}
 For large variance, the unsimplified expression~\eqref{ineq-sigma} is perhaps preferable because its second term
can be made small (see Example~\ref{rmk-historyPoisson} below). It should be noted, however, that for moderate values of the variance, the obtained bound in the simplified expression~\eqref{ineq-sigmabis} can be valuable. 
For example, when $X\sim\mathcal{B}(p)$ is a Bernoulli random variable of entropy $H_{\mathrm{b}}(p)=p\log\frac{1}{p}+(1-p)\log\frac{1}{1-p}$, the sum in~\eqref{ineq-sigmabis}  has only two terms:
\begin{equation}\label{ineq-moustache}
H_{\mathrm{b}}(p) \leq  \log \bigl( e^{\frac{\frac{1}{2}-p}{1-p}}+e^{\frac{p-\frac{1}{2}}{p}}\bigr).
\end{equation}
This is illustrated in Fig.~\ref{figmoustache}.
On the scale of the figure, when the variance is not too small ($|p-\frac{1}{2}|<0.2$), the two curves are indistinguishable, while in comparison Massey's original bound~\eqref{ineq-Massey} is much looser.
\end{remark}
\begin{figure}[htbp!]
\centering
\includegraphics[width=.7\linewidth]{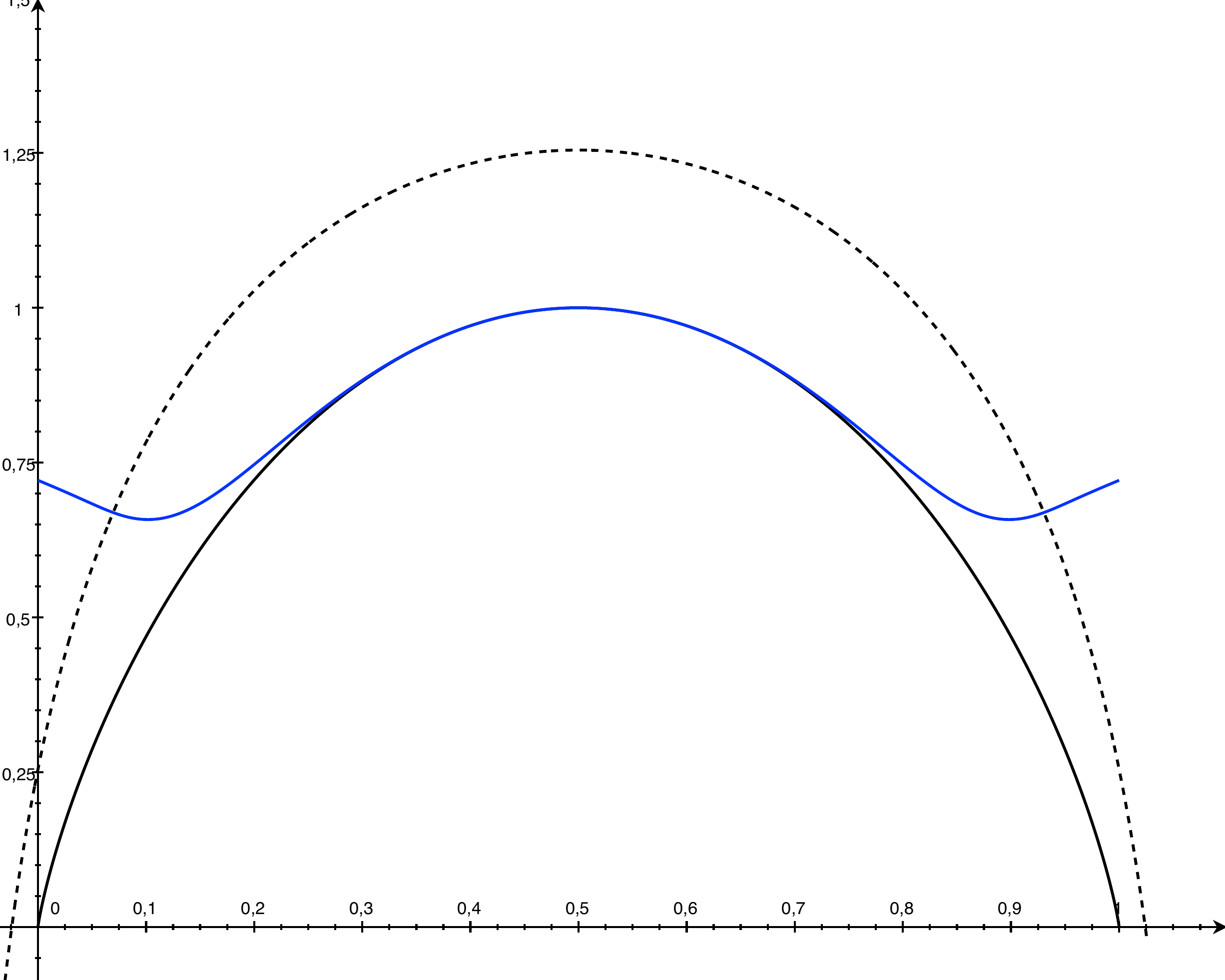}
\setlength{\unitlength}{0.55pt}
\begin{picture}(0,0)
\put(-20,20){\tiny$p$}
\put(-250,230){\tiny$\frac{1}{2}\log\bigl(2\pi e (p(1-p)+\frac{1}{12})\bigr)$}
\put(-275,115){\tiny$H_{\mathrm{b}}(p)=p\log\frac{1}{p}+(1-p)\log\frac{1}{1-p}$}
\put(-245,176){\textcolor{blue}{\footnotesize$\log \bigl( e^{\frac{\frac{1}{2}-p}{1-p}}+e^{\frac{p-\frac{1}{2}}{p}}\bigr)$}}
\end{picture}
\caption{Moustache bound~\eqref{ineq-moustache} (blue) vs. Massey's bound~\eqref{ineq-Massey} (dashed) on the binary entropy function (in bits).
} \label{figmoustache}
\end{figure}
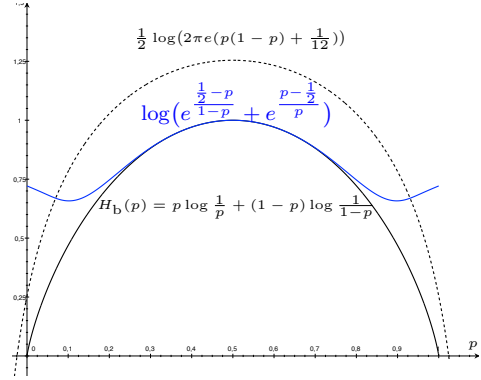%

\paragraph{Mean parameter} 

\begin{corollary} Let $X\geq 0$ be integer-valued with finite mean $\mu$. Then
\begin{equation}\label{ineq-mu}
H(X) \leq \log(e\mu) +\log \sum_x \frac{e^{-x/\mu}}{\mu} =  \log e +\log \sum_x e^{-x/\mu}
\end{equation}
the sums being taken over all nonnegative integer values $x$ of~$X$.

For $\alpha>\frac{1}{2}$ and any integer-valued $X\geq 0$ with finite mean~$\mu$, 
\begin{equation}\label{ineq-mean-alpha}
H_\alpha(X)\leq   \frac{\alpha}{1-\alpha}\log \frac{\alpha}{2\alpha-1} 
+\log\sum_x
\Bigl(1+\frac{1-\alpha}{2\alpha-1} \cdot\frac{x}{\mu}\Bigr)_{\!\!+}^{\frac{\alpha}{\alpha-1}}
\end{equation}
the sum being taken over all nonnegative integer values $x$ of~$X$.
\end{corollary}

\begin{proof}
For $\alpha=1$ we take $\mathcal{X}$ with exponential  density $\frac{e^{-x/\mu}}{\mu}$ of differential entropy $h(\mathcal{X})=\log(e\mu)$.
Theorem~\ref{thm-two} gives~\eqref{ineq-mu}.

For $\alpha\ne 1$, $\mathcal{X}$ is $\alpha$-exponential of mean $\mu_{\mathcal{X}}=\mu$ and differential entropy  given by the r.h.s. of~\eqref{maxentmean-alpha}.
From the expression of an $\alpha$-exponential~\eqref{alpha-exp}, we have
$
f_\alpha(x) = \frac{1}{Z_\alpha} \bigl(1+\beta \frac{x}{\mu}\bigr)_+^{\frac{\alpha}{\alpha-1}}
$
where $\beta=\frac{1-\alpha}{2\alpha-1}$ and $Z_\alpha$ is given by~\eqref{eq-Zalpha-mean}, that is, $Z_\alpha=\mu$.
Theorem~\ref{thm-two-alpha} gives~\eqref{ineq-mean-alpha}.
\end{proof}

\begin{remark}
Again the sum in~\eqref{ineq-mu} or~\eqref{ineq-mean-alpha} does not need to be taken over \textit{all} $x\in\N$ if the support of $X$ is limited. 
In particular, when $\alpha>1$, the sum in~\eqref{ineq-mean-alpha} is restricted to values $x$ in the interval $0\leq x < \frac{2\alpha-1}{\alpha-1}\mu$.  

If, however, the sum is to be taken over $\N$, then evaluating the geometric sum $\sum_{x\in\N} e^{-x/\mu}= \frac{1}{1-e^{-1/\mu}}$ in~\eqref{ineq-mu} gives the inequality
\begin{equation}\label{ineq-mubof}
H(X) \leq \log e -\log (1-e^{-1/\mu}),
\end{equation}
As seen in Subsection~\ref{unfortunate} below, however, this bound turns out to be always weaker then the corresponding Massey-type inequality~\eqref{ineq-rioul1}. %

\end{remark}

\subsection{Use of the Poisson Summation Formula}

When $\sigma^2$ or $\mu$ is large, then the additional logarithmic term $\log Z'$ in~\eqref{ineq-general-discrete} is likely to be small because of the approximation $Z'=\sum_x f(x) \approx \int f(x)\dx = 1$. In order to evaluate this precisely, the \emph{Poisson summation formula} can be used. 

\begin{lemma}[Poisson Summation Formula~{\cite[p.\,252]{SteinWeiss71}}]
Let $f$ be Lebesgue-integrable and let
\begin{equation}
\hat{f}(t)\triangleq \int_{-\infty}^{+\infty} f(x) \,e^{-2i\pi tx} \dx 
\end{equation}
be the Fourier transform of $f(x)$. 
If both $f$ and $\hat{f}$ have $O(\frac{1}{|x|^{1+\eps}})$ decay at infinity then Poisson's summation formula holds:
\begin{equation}\label{eq-poisson}
\sum_{x\in\Z} f(x) =\sum_{x\in\Z} \hat{f}(x)
\end{equation}
where the $x=0$ term in the r.h.s. is  $\hat{f}(0)=\int f(x)\dx=1$. 
\end{lemma}
The Fourier transform pairs used in this paper are given in Table~\ref{tab-FT}.

\begin{table}[htb!]
\centering
\caption{Some Fourier transform pairs.} \label{tab-FT}\normalsize
\renewcommand{\arraystretch}{1.3}
\newcolumntype{C}{>{$}c<{$}}
\begin{tabular}{|C|C|}\hline 
f(x) & \hat{f}(x) \\\hline\hline
\rule{0pt}{3.5ex}%
\dfrac{1}{\sqrt{2\pi\sigma^2}}e^{-\frac{1}{2}(\frac{x-\mu}{\sigma})^2}
&e^{-2i\pi\mu x} e^{-2(\pi\sigma x)^2} \\[1.5ex]\hline
\rule{0pt}{4ex}%
\dfrac{e^{-|x|/\mu}}{\mu}&\dfrac{2}{1+(2\pi\mu x)^2}
\\[1.5ex]\hline
\rule{0pt}{3.5ex}%
\dfrac{1}{\pi\sigma}\,\dfrac{1}{1+ (\frac{x-\mu}{\sigma})^2}
&
e^{-2i\pi\mu x} e^{-2\pi\sigma |x|}
\\[2ex]\hline
\rule{0pt}{3.5ex}%
\dfrac{2}{\pi\sigma}\dfrac{1}{(1+(\frac{x-\mu}{\sigma})^2)^2}
&
e^{-2i\pi\mu x}(1+2\pi \sigma |x|) e^{-2\pi\sigma|x|}
\\[2ex]\hline
\end{tabular} 
\renewcommand{\arraystretch}{1}
\end{table}

\begin{example}\label{rmk-historyPoisson}
As an example, using the first Fourier transform pair of Table~\ref{tab-FT} in Poisson's formula~\eqref{eq-poisson} one obtains\smallskip{} $\sum_{x\in\Z} \frac{e^{-\frac{1}{2}(\frac{x-\mu}{\sigma})^2}}{\sqrt{2\pi\sigma^2}}= \sum_{x\in\Z} 
e^{-2i\pi\mu x} e^{-2(\pi\sigma x)^2} =1+2\sum_{x=1}^{+\infty} e^{-2(\pi\sigma x)^2}\cos 2\pi\mu x$.
This identity  %
 is historically the very first occurence of the formula in 1823 by Poisson~\cite[Eq.~(15)]{Poisson1823} which was later generalized by other mathematicians to other Fourier transform pairs.
 It shows that for large variance, the second term inthe r.h.s. of~\eqref{ineq-sigma} is in fact exponentially small.
\end{example}

\section{Inequalities of the Massey Type}\label{massey-ineq}

In this section, we apply the techniques described in Section~\ref{sec-two} to obtains inequalities of the Massey type.
In keeping with Remark~\ref{rmk-Z}, we assume that $X$ is integer-valued, with mean $\mu$ and variance $\sigma^2$, and we apply Theorem~\ref{thm-one} in the form $H_\alpha(X)=h_\alpha(\mathcal{X})-h_\alpha(\mathcal{U})$ where $\mathcal{U}$ has support of finite length $\ell(\mathcal{U})=\Delta\leq 1$. Then 
Kullback's inequality~\eqref{ineq-general-continuous} or~\eqref{ineq-general-continuous-alpha}  applied to $\mathcal{X}=X+U$ provides various upper bounds on the discrete entropy $H(X)$ from upper bounds on $h(\mathcal{X})$.

We illustrate this approach here in the three classical situations \textit{a)}, \textit{b)}, \textit{c)} of Subsection~\ref{sec-two-ex}, where we respectively have
\begin{enumerate}[\it a)]
\item Support length $\ell(\mathcal{X})= \ell(X)+\ell(\mathcal{U})= \ell(X)+\Delta$;
\item Variance $\sigma_{\mathcal{X}}^2=\sigma^2+\sigma_{\mathcal{U}}^2$;
\item Mean $\mu_{\mathcal{X}}=\mu+\mu_{\mathcal{U}}$.
\end{enumerate}

\subsection{Inequalities for Fixed Support Length}\label{iffsl}

Suppose that $X$ has finite support $\{k,\ldots,k+\ell\}$ of length $\ell\geq 0$. Since $\ell(\mathcal{X})= \ell(X)+\ell(\mathcal{U})= \ell+\Delta$, by Theorem~\ref{thm-one} and inequality~\eqref{maxentunif} or~\eqref{maxentunif-alpha}, we have
\begin{equation}\label{ineq-triv}
H_\alpha(X) \leq \log(\ell+\Delta) - h(\mathcal{U})
\end{equation}
for any $\alpha>0$.
Since $\mathcal{U}$ has support length $\Delta\leq 1$, from~\eqref{maxentunif} or~\eqref{maxentunif-alpha} we always have $h(\mathcal{U})\leq \log \Delta\leq \log 1 = 0$ with equality iff $\mathcal{U}$ is uniformly distributed in an interval of length~$\Delta=1$.  Thus, given $\Delta$, the best upper bound in~\eqref{ineq-triv} is $\log(\ell+\Delta)-\log\Delta$, which is minimized when $\Delta$ is maximum $=1$. 
One obtains the well-known bound 
\begin{equation}
H_\alpha(X)\leq \log (\ell+1) 
\end{equation}
achieved when $X$ is equiprobable (hence $\mathcal{X}=X+\mathcal{U}$ is uniformly distributed).
\begin{remark}
Interestingly, achievability of $h(X+\mathcal{U})=\log(\ell+1)$ for $\alpha=1$ is at the basis of the analysis done in~\cite[Thm.\,1]{RioulMagossi14} on Shannon’s vs. Hartley's formula. 
\end{remark}

\subsection{Inequalities for Fixed Variance} 

Suppose that $X$ has finite variance $\sigma^2$. Since $\sigma_{\mathcal{X}}^2=\sigma^2+\sigma_{\mathcal{U}}^2$, by Theorem~\ref{thm-one} and inequality~\eqref{maxentvar}, we have
\begin{equation}\label{ineq-sigmaZ}
H(X) \leq  \tfrac{1}{2}\log\bigl(2\pi e(\sigma^2+\sigma_{\mathcal{U}}^2)\bigr) - h(\mathcal{U})
\end{equation} 
where $\mathcal{U}$ has support length $\leq 1$. 
Here the best choice of $\mathcal{U}$---the best compromise between maximum possible $h(\mathcal{U})$ and minimum possible $\sigma_{\mathcal{U}}^2$---depends on the value of $\sigma^2$.
But it can be observed that the obtained bound cannot be tight for small values of $\sigma^2$. 
Indeed when $\sigma^2=0$, $X$ is deterministic, \hbox{$H(X)=0$} and the upper bound in~\eqref{ineq-sigmaZ} becomes $\tfrac{1}{2}\log(2\pi e\sigma_{\mathcal{U}}^2)-h(\mathcal{U})$ which from~\eqref{maxentvar}  is strictly positive since $\mathcal{U}$ cannot be Gaussian when it has finite support.

Therefore, for large $\sigma^2$, the best asymptotic upper bound in~\eqref{ineq-sigmaZ} is obtained when $h(\mathcal{U})$ is maximum $=\log 1=0$. From the equality case in~\eqref{maxentunif} $\mathcal{U}$ is then uniformly distributed in an interval of length~$1$. In this case $\sigma_{\mathcal{U}}^2=\frac{1}{12}$ and one recovers \emph{Massey's inequality}~\cite{Massey88} 
\begin{equation}\label{ineq-Massey}
H(X)< \tfrac{1}{2}\log\bigl(2\pi e(\sigma^2+\tfrac{1}{12})\bigr)
\end{equation}
for any fixed $\sigma^2$, where the strictness of the inequality follows from the fact that $\mathcal{X}=X+\mathcal{U}$ is not Gaussian.

\begin{remark}
The bound~\eqref{ineq-Massey} is asymptotically tight for large~$\sigma^2$: As an example, for Poisson distributed $X$ 
we have~\cite{Evans88} $H(X)= \tfrac{1}{2}\log(2\pi e\sigma^2) + O(\frac{1}{\sigma^{2}})$.
However, it can still be improved: Section~\ref{sec-three} shows that the $\frac{1}{12}$ constant in~\eqref{ineq-Massey}  can be replaced by an arbitrary small constant as $\sigma$ gets large.
\end{remark}

The natural generalization of Massey's inequality~\eqref{ineq-Massey} to $\alpha$-entropies is given by the folllowing
\begin{theorem} For any integer-valued $X$ with finite variance $\sigma^2$,
\begin{equation}\label{ineq-Massey-alpha}
H_\alpha(X)< 
\begin{cases}
\frac{1}{2} \log \bigl(\frac{3\alpha-1}{1-\alpha}\pi(\sigma^2+\frac{1}{12})\bigr)
+\frac{1}{1-\alpha}\log \frac{2\alpha}{3\alpha-1} 
\\[1ex]\;+ \log \frac{\Gamma(\frac{1}{1-\alpha}-\frac{1}{2})}{\Gamma(\frac{1}{1-\alpha})}
\qquad\qquad\qquad\text{for $\frac{1}{3}<\alpha<1$}\\[3ex] 
\frac{1}{2} \log \bigl(\frac{3\alpha-1}{\alpha-1}\pi(\sigma^2+\frac{1}{12})\bigr)
+\frac{1}{\alpha-1}\log \frac{3\alpha-1}{2\alpha}
\\[1ex]\;+\log\frac{\Gamma(\frac{\alpha}{\alpha-1})}{\Gamma(\frac{\alpha}{\alpha-1}+\frac{1}{2})}
\qquad\qquad\qquad\text{for $\alpha>1$.}
\end{cases}
\end{equation} 
\end{theorem}

\begin{proof}
With a similar reasoning as above in the case $\alpha=1$ for large $\sigma^2$, the best upper bound in Theorem~\ref{thm-one} is obtained when $\mathcal{U}$ is uniformly distributed in an interval of length~$1$. Hence~\eqref{eq-hH==} holds, and since $\sigma_{\mathcal{X}}^2=\sigma^2+\sigma_{\mathcal{U}}^2=\sigma^2+\frac{1}{12}$, \eqref{maxentvar-alpha} gives~\eqref{ineq-Massey-alpha}.
The strictness of the inequality follows from the fact that $\mathcal{X}=X+\mathcal{U}$ (which has a staircase density) cannot be $\alpha$-Gaussian.
\end{proof}

\begin{example}
Thus, referring to Example~\ref{ex-alpha-sigma},
\begin{align}
\label{ex-1/2-sigma}
H_{\frac{1}{2}}({X})&< \frac{1}{2}\log\Bigl(4\pi^2\bigl(\sigma^2+\frac{1}{12}\bigr)\Bigr)\\[1ex]
\label{ex-2/3-sigma}
H_{\frac{2}{3}}({X})&<  \frac{1}{2}\log\Bigl(\frac{64}{27} \pi^2\bigl(\sigma^2+\frac{1}{12}\bigr)\Bigr) \\[1ex]
H_{2}({X})&< 
\frac{1}{2}\log\Bigl(\frac{125}{9}\bigl(\sigma^2+\frac{1}{12}\bigr)\Bigr)\\[1ex]
H_{3}({X})&< 
\frac{1}{2}\log\Bigl(\frac{4}{3}\pi^2\bigl(\sigma^2+\frac{1}{12}\bigr)\Bigr).
\end{align}
\end{example}

\begin{remark}
Such inequalities cannot exist in general when $\alpha\leq\frac{1}{3}$.  To see this, consider the discrete random variable $X\geq 1$ having distribution $\P(X=k)=\frac{c}{(k\log k)^3}$ with normalization constant $c=\sum_{k>0} \frac{1}{(k\log k)^3}$. Then $X$ has finite second moment $\sum_{k>0} \frac{c}{k\log^3 k}<+\infty$ hence finite variance, but $\sum_{k>0}\sqrt[3]{\P(X=k)}= \sum_{k>0} \frac{1}{k\log k}=+\infty$, hence $H_\alpha(X) \geq H_{\frac{1}{3}}(X)=+\infty$ for all $\alpha\leq\frac{1}{3}$.
\end{remark}

\subsection{Inequalities for Fixed Mean} \label{iffm}

Suppose that $X\geq 0$ has finite mean $\mu$. Since $\mu_{\mathcal{X}}=\mu+\mu_{\mathcal{U}}$, by Theorem~\ref{thm-one} and inequality~\eqref{maxentmean}, we have
\begin{equation}\label{ineq-muZ}
H(X) \leq \log\bigl(e(\mu+ \mu_\mathcal{U})\bigr) - h(\mathcal{U})
\end{equation}
provided that $\mathcal{U} \geq 0$ a.s. with support length $\leq 1$.

Again the best choice of $\mathcal{U}$ (the best compromise between maximum possible $h(\mathcal{U})$ and minimum possible $\mu_{\mathcal{U}}$) depends on the value of the parameter $\mu\geq 0$. Also 
the obtained bound cannot be tight for small values of $\mu$: When $\mu=0$, \hbox{$X=0$}~a.s., $H(X)=0$ and the upper bound in~\eqref{ineq-muZ} becomes $\log\bigl(e\mu_\mathcal{U}\bigr) - h(\mathcal{U})$ which from~\eqref{maxentmean}  is strictly positive because $\mathcal{U}$ cannot be exponential when it has finite support. 

For large $\mu$,  the best asymptotic upper bound in~\eqref{ineq-muZ} is again obtained when $h(\mathcal{U})$ is maximum $=\log 1=0$. 
From the equality case in~\eqref{maxentunif} $\mathcal{U}\geq 0$ is then uniformly distributed in an interval of length~$1$. In this case the minimum value of $\mu_\mathcal{U}$ is achieved when $\mathcal{U}\geq 0$ is uniformly distributed in $(0,1)$, which gives $\mu_\mathcal{U}=\frac{1}{2}$.
We obtain the following variation of Massey inequality.

\begin{theorem}
For any integer-valued $X\geq 0$ with finite mean~$\mu$,
\begin{equation}\label{ineq-rioul1}
H(X) < \log\bigl(e(\mu+ \tfrac{1}{2})\bigr)
\end{equation} 
\end{theorem}

Here the strictness of the inequality follows from the fact that $\mathcal{X}=X+\mathcal{U}$ is not exponential, hence~\eqref{maxentmean} cannot be achieved with equality.

\begin{remark}
The bound~\eqref{ineq-rioul1} is asymptotically tight for large~$\mu$: As an example, for geometric $X$ 
we have $H(X)=\mu H_{\mathrm{b}}(1/\mu)=\log (e\mu) + O(\frac{1}{\mu})$ where $H_{\mathrm{b}}(p)=p\log\frac{1}{p}+(1-p)\log\frac{1}{1-p}$ is the binary entropy function.
\end{remark}

The natural generalization of~\eqref{ineq-rioul1} to $\alpha$-entropies is given by the following
\begin{theorem}
For any integer-valued $X\geq 0$ with mean~$\mu$ and any $\alpha>\frac{1}{2}$,
 \begin{equation}\label{ineq-rioul1-alpha} 
\begin{aligned}
H_\alpha(X) <\; &\log(\mu+\frac{1}{2}) + \frac{\alpha}{1-\alpha} \log \frac{\alpha}{2\alpha-1} 
\\
&=\log(\mu+\frac{1}{2}) + \frac{\alpha}{\alpha-1} \log \frac{2\alpha-1}{\alpha}. 
\end{aligned}
\end{equation}
\end{theorem}

\begin{proof}
For large $\mu$, as in the case $\alpha=1$ above, the best upper bound in Theorem~\ref{thm-one} is obtained when $\mathcal{U}$ is uniformly distributed in $(0,1)$. 
Hence~\eqref{eq-hH==} holds, and since $\mu_{\mathcal{X}}=\mu+\mu_{\mathcal{U}}=\mu+\frac{1}{2}$, 
\eqref{maxentmean-alpha} gives~\eqref{ineq-rioul1-alpha}
for any $\alpha>\frac{1}{2}$, where the strictness of the inequality follows from the fact that $\mathcal{X}=X+\mathcal{U}$ (which has a staircase density) cannot be $\alpha$-exponential. 
\end{proof}

\begin{example}\label{ex-alpha-mu+1/2}
Thus, referring to Example~\ref{ex-alpha-mu},
\begin{align}
H_{\frac{2}{3}}({X})&< \log (4\mu+2)\\
H_{\frac{3}{4}}({X})&< \log\frac{27(\mu+\frac{1}{2})}{8}\\
H_{2}({X})&< \log\frac{9(\mu+\frac{1}{2})}{4}.
\end{align}
\end{example}
\begin{remark}\label{rmk-nogo}
Such inequalities cannot exist in general when $\alpha\leq\frac{1}{2}$.  To see this, consider the discrete random variable $X\geq 1$ with distribution $\P(X=k)=\frac{c}{(k\log k)^2}$ where  $c=\sum_{k>0} \frac{1}{(k\log k)^2}$ is a normalization constant. Then $X$ has finite mean $\mu= \sum_{k>0} \frac{c}{k\log^2 k}<+\infty$ but $\sum_{k>0}\sqrt{\P(X=k)}= \sum_{k>0} \frac{1}{k\log k}=+\infty$, hence $H_\alpha(X) \geq H_{\frac{1}{2}}(X)=+\infty$ for all $\alpha\leq\frac{1}{2}$.
\end{remark}

\section{Improved Inequalities}\label{sec-three}

In this section, we apply the alternative bounding techniques described in Section~\ref{sec-alt} with the aim to improve the previous inequalities of the Massey type.
Applying Theorem~\ref{thm-two} or~\ref{thm-two-alpha} will have the effect of removing the constant $\frac{1}{12}$ in~\eqref{ineq-Massey} and $\frac{1}{2}$ in~\eqref{ineq-rioul1} at the expense of an additional additive constant $\log Z'$ or $\log Z'_\alpha$ in the upper bound. 

We again consider an integer-valued variable under the three classical situations \textit{a)}, \textit{b)}, \textit{c)} of Subsection~\ref{subsec-mixed-ex}.

\subsection{Inequalities for Fixed Support Length}

In case~ \textit{a)} we have already seen in Subsection~\ref{subsec-mixed-ex} that one obtains the known inequality
$H_\alpha(X)\leq \log(\ell+1)$ achieved when $X$ of support length $\ell$ is equiprobable. Thus in this case,  no improvement is possible: We obtain the same result as in Subsection~\ref{iffsl}.

\subsection{Inequalities for Fixed Mean} \label{unfortunate}

Here we assume $X\geq 0$ with fixed mean $\mu$. For $\alpha=1$, inequality~\eqref{ineq-mu} applies with \smallskip $Z'=\frac{1}{\mu}\sum_{x\in\N} e^{-x/\mu}$.
Using the second Fourier transform pair of Table~\ref{tab-FT} in Poisson's formula~\eqref{eq-poisson}
we obtain  $\sum_{x\in\Z} \frac{e^{-|x|/\mu}}{\mu} = \sum_{x\in\Z} \frac{2}{1+(2\pi\mu x)^2}$, which gives
\begin{equation}\label{eq-poisson-laplacian}
Z'=\frac{1}{\mu}\sum_{x\in\N} e^{-x/\mu} = 1+\frac{1}{2\mu}+2\sum_{x=1}^{+\infty} \frac{1}{1+(2\pi\mu x)^2}.
\end{equation}
Here we have applied Poisson's formula to the symmetrized density $\frac{1}{2}\bigl(f(x)+f(-x)\bigr)$ to ensure that the decay condition at infinity holds for the Fourier transform.
It follows from~\eqref{eq-poisson-laplacian} that
\begin{equation}
\sum_{x\in\N} \frac{e^{-x/\mu}}{\mu} > 1+\frac{1}{2\mu},
\end{equation}
which implies that~\eqref{ineq-mubof} is strictly weaker than the Massey-type inequality~\eqref{ineq-rioul1}: %
In fact, \eqref{ineq-rioul1} already reads 
$H(X) < \log(e\mu) +\log (1+\frac{1}{2\mu})$. 

A similar phenomenon occurs when $\alpha\ne 1$. In fact, comparing 
\eqref{ineq-mean-alpha} to~\eqref{ineq-rioul1-alpha} one has
\begin{equation}\label{trapezoidal}
\sum_{x\in\N}\Bigl(1+\frac{1-\alpha}{2\alpha-1} \cdot\frac{x}{\mu}\Bigr)_{\!\!+}^{\frac{\alpha}{\alpha-1}} > \mu+\frac{1}{2}. 
\end{equation}
for any $\alpha>1/2$
(See Appendix~\ref{app-trapezoidal} for a simple proof).
Therefore, unfortunately, the approach of this section cannot improve the result in Subsection~\ref{iffm}.

\subsection{Improved Inequalities for Fixed Variance} 

For large variance $\sigma^2$, Massey's original inequality~\eqref{ineq-Massey} reads $H(X)\leq \tfrac{1}{2}\log\bigl(2\pi e(\sigma^2+\tfrac{1}{12})\bigr)<\frac{1}{2}\log(2\pi e \sigma^2) +\frac{\log e}{24\sigma^2}$.
Now~\eqref{ineq-sigma} together with Poisson's formula~\eqref{eq-poisson} greatly improves Massey's inequality, since the $O(\frac{1}{\sigma^2})$ term can be replaced by  the exponentially small  $O(e^{-2\pi^2\sigma^2})$:
\begin{theorem}\label{thm-rioul2}
For any integer-valued $X$ of variance $\sigma^2>0$,
\begin{equation}\label{eq-rioul2}
H(X) < \frac{1}{2}\log(2\pi e \sigma^2) + \frac{2\log e}{e^{2\pi^2\sigma^2}-1}.
\end{equation}
\end{theorem}
\begin{proof}
Using the first Fourier transform pair of Table~\ref{tab-FT} in Poisson's formula~\eqref{eq-poisson} one  obtains 
\begin{equation}\label{eq-poisson-gauss}
\frac{1}{\sqrt{2\pi\sigma^2}}\sum_{x\in\Z} e^{-\frac{1}{2}(\frac{x-\mu}{\sigma})^2}
=1+2\sum_{x=1}^{+\infty} e^{-2(\pi\sigma x)^2}\cos 2\pi\mu x
\end{equation} 
The sum in the r.h.s. is bounded by
$\sum_{x\geq 1} e^{-2(\pi\sigma x)^2} \leq \sum_{x\geq 1} e^{-2(\pi\sigma)^2 x}  =\frac{1}{e^{2\pi^2\sigma^2}-1}$.  Substituting in~\eqref{ineq-sigma} and using the inequality $\log(1+z)< (\log e)z$ (when $z>0$) gives the result. 
\end{proof}

\begin{example}
As a illustration, consider a binomial $X\sim\mathcal{B}(n,p)$ of variance $\sigma^2 = npq$ (where $p+q=1$). The best known upper bound on $H(X)$ is~\cite[Eq.\,(7)]{Adell10}
\begin{equation}
H(X)   < \frac{1}{2}\log(2\pi enpq) + \frac{\log e}{12n} + \frac{\log(pq)}{2n} + \frac{\log e}{6npq}
\end{equation}
which~\eqref{eq-rioul2} considerably improves for large $n$ since all $O(\frac{1}{n})$ terms are replaced by $O(e^{-2\pi^2 npq})$:
\begin{equation}
H(X)   < \frac{1}{2}\log(2\pi enpq) + \frac{2\log e}{e^{2\pi^2npq}-1} .
\end{equation} 
\end{example}

The exponentially small term can even be made disappear under mild conditions.
For example:
\begin{corollary}\label{cor-three}
If the integer-valued variable $X\in\N$ is nonnegative and $\frac{\mu}{\sigma^2}$ is bounded by a constant $<2\pi$, then for large enough $\sigma^2$,
\begin{equation}\label{eq-rioul3}
H(X) < \frac{1}{2}\log(2\pi e \sigma^2).
\end{equation}
\end{corollary}
\begin{proof}
Apply~\eqref{ineq-sigma} where the sum can be taken only over $x\in\N$. Then by~\eqref{eq-poisson-gauss}, 
\begin{equation*}
\sum_{x\in\N} \frac{e^{-\frac{1}{2}(\frac{x-\mu}{\sigma})^2}}{\sqrt{2\pi\sigma^2}}
\leq 1+2\sum_{x=1}^{+\infty} e^{-2(\pi\sigma x)^2}%
 - \sum_{x=1}^{+\infty} \frac{e^{-\frac{1}{2}(\frac{x+\mu}{\sigma})^2}}{\sqrt{2\pi\sigma^2}}.
\end{equation*}
To obtain~\eqref{eq-rioul3} it is sufficient to prove that 
$2e^{-2(\pi\sigma x)^2} < \frac{e^{-\frac{1}{2}(\frac{x+\mu}{\sigma})^2}}{\sqrt{2\pi\sigma^2}}$,
i.e., $2(\pi\sigma x)^2 - \frac{1}{2}(\frac{x+\mu}{\sigma})^2>\log\sqrt{8\pi\sigma^2}$
for all \hbox{$x\geq 1$}. When $2\pi\sigma^2>1$ we have $2(\pi\sigma)^2 > 1/2\sigma^2$ and it is enough to prove the required inequality for $x=1$, i.e., 
$(2\pi\sigma)^2 > (\frac{\mu+1}{\sigma})^2+\log({8\pi\sigma^2})$. This will hold for large enough $\sigma^2$ provided that $2\pi\sigma^2>(1+\eps)\mu$ for some $\eps>0$.
\end{proof}

\begin{example}
As an example, if $X\sim \mathcal{P}(\lambda)$ is Poisson-distributed then $\frac{\mu}{\sigma^2}=\frac{\lambda}{\lambda}=1<2\pi$ so that for large enough $\lambda$,
\begin{equation}
H(X) < \frac{1}{2}\log (2\pi e \lambda).
\end{equation}
It is found numerically that this inequality holds as soon as $\lambda>0.1312642451\ldots$. 
\end{example}

\begin{example}
Similarly, if $X\sim\mathcal{B}(n,p)$ is binomial, we may always assume that $p\leq \frac{1}{2}$ since considering $n-X$ in place of $X$ permutes the roles of $p$ and $q=1-p$ without changing $H(X)$. Then $\frac{\mu}{\sigma^2}=\frac{np}{npq}=\frac{1}{q}\leq 2 <2\pi$, and by Corollary~\ref{cor-three}, for large enough $n$,
\begin{equation}
H(X) < \frac{1}{2}\log (2\pi e npq).
\end{equation}
It is found numerically that this inequality holds for all $n>0$
as soon as $|p-\frac{1}{2}|<0.304449\ldots$.
\end{example}

\begin{remark}
For the last two examples, %
Takano's strong central limit theorem~\cite[Thm.~2]{Takano87} implies that 
\begin{equation}
H(X) = \frac{1}{2}\log(2\pi e \sigma^2) + o\Bigl( \frac{1}{\sigma^{1+\eps}}\Bigr) 
\end{equation}
for every $\eps>0$. The above inequalities show that the $o\bigl( \frac{1}{\sigma^{1+\eps}}\bigr)$ term is actually negative for large enough $\sigma$. 
\end{remark}

We now illustrate the use of the Poisson summation formula~\eqref{eq-poisson} in~\eqref{ineq-sigma-alpha}
for $\alpha$-entropies, in the two cases $\alpha=\frac{1}{2}$ and $\alpha=\frac{2}{3}$.

\begin{lemma}\label{lem-1/2-2/3}
One has the following Poisson summation formulas:
\begin{equation}\label{Z'1/2}
Z'_\frac{1}{2}=  \tfrac{1}{\pi\sigma}\!\sum_{x\in\Z}\tfrac{1}{1+ (\frac{x-\mu}{\sigma})^2}=1+2\sum_{x=1}^{+\infty} e^{-2\pi\sigma x}\cos 2\pi\mu x.
\end{equation}
\begin{equation}\label{Z'2/3}
Z'_\frac{2}{3}\!=\! \tfrac{2}{\pi\sigma}\!\sum_{x\in\Z}\!\!\tfrac{1}{(1+(\frac{x-\mu}{\sigma})^2)^2} \!=\!
1+2\!\sum_{x=1}^{+\infty} \!(1+2\pi \sigma x) e^{-2\pi\sigma x}\!\cos 2\pi\mu x.
\end{equation}
\end{lemma}

\begin{proof}
By~\eqref{alpha-Gaussian} the $\frac{1}{2}$-Gaussian density is of the form 
$f(x)=\frac{1}{Z}(1+ (\frac{x-\mu}{\sigma})^2)^{-2}$. %
It follows that $f_\frac12(x)=\frac{1}{Z_\alpha}(1+ (\frac{x-\mu}{\sigma})^2)^{-1}=\frac{1}{\pi\sigma}\frac{1}{1+ (\frac{x-\mu}{\sigma})^2}$. \smallskip{}Using the third Fourier transform pair of Table~\ref{tab-FT} in Poisson's formula~\eqref{eq-poisson} one obtains $\sum_{x\in\Z} \frac{1}{\pi\sigma}\frac{1}{1+ (\frac{x-\mu}{\sigma})^2}=\sum_{x\in\Z} e^{-2i\pi\mu x} e^{-2\pi\sigma |x|}$,
\smallskip{}which is~\eqref{Z'1/2}.

By~\eqref{alpha-Gaussian} the $\frac{2}{3}$-Gaussian density is of the\smallskip{} form 
$f(x)=\frac{1}{Z}(1+ \beta(\frac{x-\mu}{\sigma})^2)^{-3}$ where $\beta=\frac{1}{3}$. 
It follows that\smallskip{} $f_\frac23(x)=\frac{1}{Z_\alpha} (1+ \beta(\frac{x-\mu}{\sigma})^2)^{-2}=\frac{2}{\pi\sigma}\frac{1}{(1+(\frac{x-\mu}{\sigma})^2)^2}$. \smallskip{}Using the fourth Fourier transform pair of Table~\ref{tab-FT} in Poisson's formula~\eqref{eq-poisson} one obtains\smallskip{} 
$\sum_{x\in\Z}\frac{2}{\pi\sigma}\frac{1}{(1+(\frac{x-\mu}{\sigma})^2)^2}=\sum_{x\in\Z} e^{-2i\pi\mu x}(1+2\pi \sigma |x|) e^{-2\pi\sigma|x|}$\smallskip{}, %
which is~\eqref{Z'2/3}.
\end{proof}

In the two cases $\alpha=\frac{1}{2}$ and $\frac{2}{3}$, %
the Massey-type inequalities~\eqref{ex-1/2-sigma} and~\eqref{ex-2/3-sigma} write $H_{\frac{1}{2}}(X)\leq  \frac{1}{2}\log\bigl(4\pi^2(\sigma^2+\frac{1}{12})\bigr)<\log(2\pi \sigma) +\frac{\log e}{24\sigma^2}$ and
$
H_{\frac{2}{3}}(X)\leq   \frac{1}{2}\log\Bigl(\frac{64}{27} \pi^2\bigl(\sigma^2+\frac{1}{12}\bigr)\Bigr)
<\log(\frac{8}{3\sqrt{3}} \pi\sigma) +\frac{\log e}{24\sigma^2}$,
respectively. In these inequalities, the $O(\frac{1}{\sigma^2})$ term can be replaced by  the exponentially small  $O(e^{-2\pi\sigma})$ and $O(\sigma e^{-2\pi\sigma})$, respectively:

\begin{theorem}\label{cor-1/2-2/3}
For any integer-valued $X$ of variance $\sigma^2>0$,
\begin{align}
H_{\frac{1}{2}}(X) &< \log(2\pi \sigma) + \frac{2\log e}{e^{2\pi\sigma}-1}
\label{ineq-rioul2-1/2}\\
H_{\frac{2}{3}}(X) &< \log\bigl(\frac{8 \pi\sigma}{3\sqrt{3}} \bigr)+ \frac{4(1+\pi\sigma)\log e}{e^{2\pi\sigma}-1}.
\label{ineq-rioul2-2/3}
\end{align}
\end{theorem}
\begin{proof}
The sum in the r.h.s. of~\eqref{Z'1/2} is bounded by
$\sum_{x\geq 1} e^{-2\pi\sigma x} = \frac{1}{e^{2\pi\sigma}-1}$.  Substituting in~\eqref{ineq-sigma-alpha} and using the inequality $\log(1+z)< (\log e)z$ (when $z>0$) gives~\eqref{ineq-rioul2-1/2}. 

Likewise, the sum in the r.h.s. of~\eqref{Z'2/3} is bounded by
$\sum_{x\geq 1} (1+2\pi \sigma x)e^{-2\pi\sigma x} 
=\frac{1+2\pi\sigma}{e^{2\pi\sigma}-1} + \frac{2\pi\sigma}{(e^{2\pi\sigma}-1)^2}
<2\frac{1+\pi\sigma}{e^{2\pi\sigma}-1}
$ (where we used that $2\pi\sigma<e^{2\pi\sigma}-1$).
Substituting in~\eqref{ineq-sigma-alpha} and using the inequality $\log(1+z)< (\log e)z$ (when $z>0$) gives~\eqref{ineq-rioul2-2/3}. 
\end{proof}

\begin{remark}
Using the Poisson summation formula on other Fourier transform pairs, it is possible to generalize Theorem~\ref{cor-1/2-2/3} to any value of the form $\alpha=\frac{k+1}{k+2}$ ($k=0,1,\ldots$) and prove that
\begin{equation}
H_{\frac{k+1}{k+2}}(X) <  \log(c_k\pi\sigma) +O(\sigma^ke^{-2\pi\sigma})
\end{equation}
where the constant $c_k$ is given by
\begin{equation}
c_k={ 4\sqrt{2k+1}} \binom{2k}{k} \biggl(\!\frac{k+1}{2(2k+1)}\!\biggr)^{k+1}. 
\end{equation}
The method of this and the previous section is not easily applicable to many other cases, however, since it
depends on the availability of simple expressions of Fourier transform pairs with sufficient decay at infinity.  
\end{remark}

\section{Application to Guessing}\label{sec-guess}

\subsection{Improved Massey's Inequality for Guessing}\label{ImprovedMasseyGuessing}

Inequality~\eqref{ineq-rioul1} can be thought of as an improvement of Massey's inequality for the guessing entropy~\cite{Massey94}. To see this,
let $\truemathcal{G}(X)$ be the number of successive guesses of some (discrete valued) secret~$X$ before the actual value of $X$ is found, and define the \emph{guessing entropy} as the minimum average number of guesses for a given probability distribution of $X$:
\begin{equation}\label{guessingentropy}
G(X) \triangleq \min \E\bigl(\truemathcal{G}(X)\bigr).
\end{equation}
Massey's original inequality reads~\cite{Massey94}
\begin{equation}\label{eq-Massey-guess-orig}
 G(X)\geq 2^{H(X)-2}+1 \text{ when $H(X)\geq 2$ bits.}
\end{equation}
A more general situation described by Arikan in~\cite{Arikan96} is when one guesses $X$ given the observed output $Y$ of some side channel.  The corresponding (conditional) guessing entropy is~\cite{Arikan96}
\begin{equation}\label{guessingentropycond}
G(X|Y) \triangleq \E\bigl(G(X|Y=y)\bigr)
\end{equation}
where the expectation is over $Y$'s distribution.

\begin{theorem}[Improvement of Massey's Inequality]\label{cor-one}
When $H(X)$ or $H(X|Y)$ is expressed in bits,
\begin{align}\label{ineq-rioul2}
G(X) &> \frac{2^{H(X)}}{e} + \frac{1}{2}.
\\\label{ineq-rioul2|y}
G(X|Y) &> \frac{2^{H(X|Y)}}{e} + \frac{1}{2}.
\end{align}
\end{theorem}

\begin{proof}
As explained in~\cite{Massey94} the optimal strategy leading to the minimum~\eqref{guessingentropy} require $k$ guesses with probability 
\begin{equation}
\P(\mathcal{G}(X)=k)=p_{(k)} \qquad (k=1,2,\ldots)
\end{equation}
where $p_{(k)}$ is the $k$th largest probability in $X$'s distribution.
Applying~\eqref{ineq-rioul1} to $\mathcal{G}(X)-1\geq 0$, and noting that $\mu=G(X)-1$ and $H(\mathcal{G}(X))=H(X)$ yields
\begin{equation}
H(X) < \log\bigl(e(G(X)- \tfrac{1}{2})\bigr) 
\end{equation}
which is~\eqref{ineq-rioul2}. 
Applying~\eqref{ineq-rioul2} to $X|Y=y$ for every $y$, taking the expectation over $Y$'s distribution and applying Jensen's inequality to the exponential function gives~\eqref{ineq-rioul2|y}.
\end{proof}

\begin{figure}[h!]
\centering
\includegraphics[width=.5\linewidth]{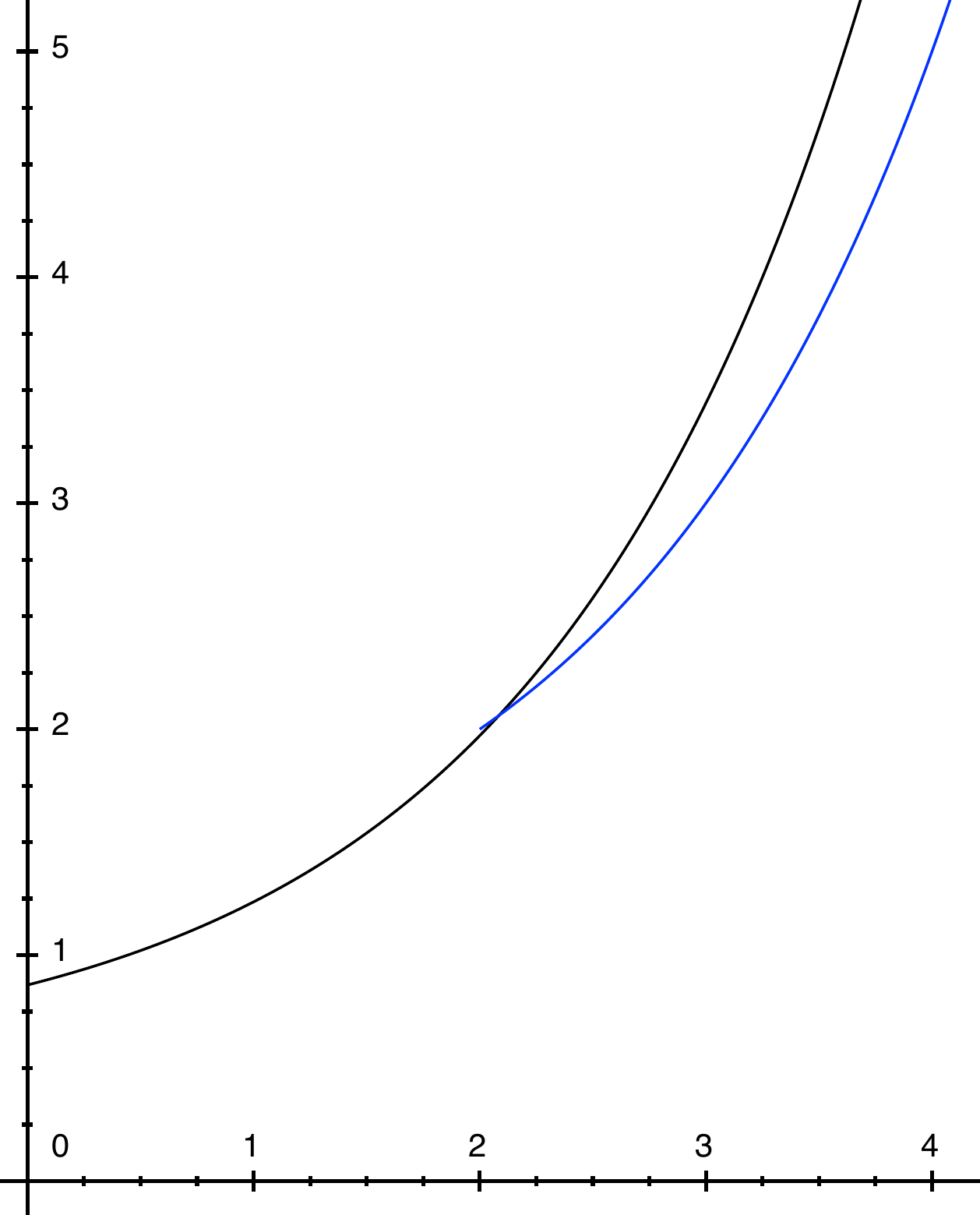}
\setlength{\unitlength}{.56pt}
\begin{picture}(0,0)
\put(-210,275){\footnotesize$\E[G(X|Y)]$}
\put(-30,30){\footnotesize$H(X|Y)$}
\put(-170,170){$\frac{2^{H(X|Y)}}{e} \!+\!\frac{1}{2}$}
\put(-95,130){\textcolor{blue}{$_{2^{H(X|Y)-2} + 1}$}}
\end{picture}
\caption{Massey's original (\textcolor{blue}{blue}) and improved (black) lower bounds.} \label{fig}
\end{figure}
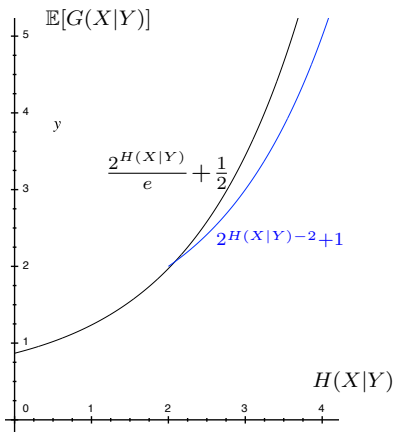

\begin{remark}
Inequality~\eqref{ineq-rioul2} improves Massey's original inequality~\eqref{eq-Massey-guess-orig} as soon as 
$H(X)  \geq  \log\frac{2e}{4-e} \approx 2.0846\ldots$ bits
and is also valid for $H(X)< 2$ bits. %
Fig.~\ref{fig} shows that the improvement over Massey's original inequality is particularly important for large values of entropy, by the factor ${4}/{e}$. %
It is quite startling to notice that the approach followed by Massey back in the 1970s~\cite{Massey88} can improve the result of his 1994 paper~\cite{Massey94} so much.

Massey's inequality was already improved by the author, with a very different proof, in the (weaker) form $G(X|Y) > \frac{2^{H(X|Y)}}{e}$, see~\cite{Rioul19} and~\cite{Popescu20}. See also~\cite{Popescu19,Popescu20} for a different kind of improvement. 

Inequality~\eqref{ineq-rioul2} or~\eqref{ineq-rioul2|y} can be shown to be the best among all possible bounds of the form $G> a\cdot b^H + c$~\cite{TanasescuChoudaryRioulPopescu21}. In particular, for large values of entropy, the gain factor $\frac{4}{e}$ of~\eqref{ineq-rioul2} over~\eqref{eq-Massey-guess-orig} is optimal, as well as the additive constant $\frac{1}{2}$.
\end{remark}

\subsection{Generalization to Rényi entropies}

In this Subsection, we consider Rényi's entropy $H_\alpha(X)$ as well as 
Arimoto's conditional entropy $H_\alpha(X|Y)$~\cite{Arimoto75,FehrBerens14} of order $\alpha>0$ which finds natural application to guessing with side information~\cite{Arikan96,SasonVerdu18,SasonVerdu18bis}. 

\begin{theorem}\label{cor-ineq-rioul2-alpha}
When $H_\alpha(X)$ and $H_\alpha(X|Y)$ are expressed in bits, for any $\alpha>\frac{1}{2}$,
\begin{align}\label{ineq-rioul2-alpha}
G(X) &> \frac{2^{H_\alpha(X)}}{(1\!+\!\frac{\alpha-1}{\alpha})^\frac{\alpha}{\alpha-1}} +\frac{1}{2}
=(1\!-\!\tfrac{1-\alpha}{\alpha})^\frac{\alpha}{1-\alpha}\!\cdot\! 2^{H_\alpha(X)}+\frac{1}{2}.
\\
\label{ineq-rioul2-alpha|y}
G(X|Y)\! &>\! \frac{2^{H_\alpha(X|Y)}}{(1\!+\!\frac{\alpha-1}{\alpha})^\frac{\alpha}{\alpha-1}} \!+\!\frac{1}{2}
=(1\!-\!\tfrac{1-\alpha}{\alpha})^\frac{\alpha}{1-\alpha}\!\cdot\! 2^{H_\alpha(X|Y)}\!+\!\frac{1}{2}.
 \end{align}
\end{theorem}

\begin{proof}
Similarly as in the preceding Subsection~\ref{ImprovedMasseyGuessing}, 
the $\mu+\frac{1}{2}$ term in~\eqref{ineq-rioul1-alpha} is replaced by $G(X)-\frac{1}{2}$, and one immediately obtains~\eqref{ineq-rioul2-alpha}.

Arimoto's conditional $\alpha$-entropy~\cite{Arimoto75} satisfies %
$H_\alpha(X|Y) 
=  \frac{\alpha}{1-\alpha} \log \E \exp \tfrac{1-\alpha}{\alpha} H_\alpha(X|Y=y)$.
Thus if $H_\alpha(X|Y)$ is expressed in bits, one has
\begin{equation}\label{def-arimoto}
 2^{H_\alpha(X|Y)} = \Bigl(\E\, 2^{\frac{1-\alpha}{\alpha}H_\alpha(X|Y=y)}\Bigr)^{\frac{\alpha}{1-\alpha}}
\end{equation}
where the expectation is over $Y$'s distribution.
Applying~\eqref{ineq-rioul2-alpha} to $X|Y=y$ for every $y$, taking the expectation over $Y$'s distribution and applying Jensen's inequality to the function $x\mapsto x^{\frac{\alpha}{1-\alpha}}$, which is strictly convex when $\alpha>\frac{1}{2}$, gives~\eqref{ineq-rioul2-alpha|y}.
\end{proof}

\begin{remark}
Since the factor $(1+\frac{\alpha-1}{\alpha})^\frac{\alpha}{\alpha-1}$ converges to $e$ as $\alpha\to 1$, Theorem~\ref{cor-one} is recovered by letting $\alpha\to 1$.
This factor is nonincreasing in $\alpha$, and since $1+x<e^x$ for $x\ne 0$, the term $(1+\frac{\alpha-1}{\alpha})^\frac{\alpha}{\alpha-1}=\frac{1}{(1-\frac{1-\alpha}{\alpha})^\frac{\alpha}{1-\alpha}}$ is greater\smallskip{} than $e$ for $\alpha<1$ and less than $e$ for $\alpha>1$.
Since $H_\alpha(X)$ is also nonincreasing in $\alpha$, none of the inequalities~\eqref{ineq-rioul2-alpha} (or~\eqref{ineq-rioul2-alpha|y}) is a trivial consequence of another for a different value of $\alpha$.
\end{remark}

\begin{example}\label{ex-alpha-G}
Thus, referring to Example~\ref{ex-alpha-mu+1/2},
\begin{align}
G(X) & >  \frac{1}{4}2^{H_{\frac{2}{3}}({X})}+\frac{1}{2}\\[1ex]
G(X) & >\frac{8}{27}2^{H_{\frac{3}{4}}({X})}+\frac{1}{2}\\[1ex]
G(X) & >\frac{4}{9}2^{H_{2}({X})}+\frac{1}{2}
\end{align}
and similarly for $X|Y$, 
where $\frac{9}{4}<e<\frac{27}{8}<4$. 
\end{example}

\begin{remark}\label{rmk-13}
By Remark~\ref{rmk-nogo}, no inequality of the type~\eqref{ineq-rioul2-alpha} or~\eqref{ineq-rioul2-alpha|y} can generally hold for $\alpha\leq \frac{1}{2}$.  This does not contradict Arikan's inequality~\cite{Arikan96} for the limiting case $\alpha=\frac{1}{2}$, which reads
\begin{equation}\label{ineq-arikan}
G(X|Y) \geq  \frac{2^{H_{\frac{1}{2}}\!(X|Y)}}{1\!+\!\ln M},
\end{equation}
because it was established when $X$ takes a \emph{finite} number $M$ of possible values. As $M\to+\infty$ the r.h.s. vanishes.
In other words, it is impossible to improve Arikan's inequality~\eqref{ineq-arikan} with some positive constant independent of $M$.
\end{remark}

\subsection{Arikan-type Inequalities for Rényi Entropies of Small Orders}\label{atifresm}

By Remark~\ref{rmk-nogo} and \ref{rmk-13}, the results of the previous subsection cannot generalize to $\alpha\leq\frac{1}{2}$. However, when $X$ takes values in a \emph{finite} alphabet of size $M$,
Arikan's inequality~\eqref{ineq-arikan}  for $\alpha=\frac{1}{2}$ and extensions of it for $\alpha<\frac{1}{2}$ can still be obtained using Theorem~\ref{thm-one} (equation~\eqref{eq-hH==}) applied to $\mathcal{G}(X)$, on top of the $\alpha$-Kullback inequality (Theorem~\ref{alpha-kullback}). In this case the density~\eqref{eq-phi-alpha} has to be constrained in a interval of finite length which depends on~$M$.

A derivation is as follows. 
Recall that $\mathcal{G}(X)\geq 1$ has mean $G(X)$ and $\alpha$-entropy $H_\alpha(X)$.
For simplicity consider $\mathcal{U}$ to be zero-mean, uniformly distributed in $(-\frac{1}{2},\frac{1}{2})$, so that $\mathcal{X}$ has the same mean $G(X)$ and is supported in the interval $(\frac{1}{2}, M+\frac{1}{2})$. Now consider
\begin{equation}\label{suboptimalphi}
\phi(x) = \frac{x^{\frac{1}{\alpha-1}} }{Z}
\end{equation}
restricted in the same interval $(\frac{1}{2}, M+\frac{1}{2})$. Then~\eqref{ineq-general-continuous-alpha} gives $H_\alpha(X)=h_\alpha(\mathcal{X}) < \frac{\alpha}{1-\alpha} \log G(X) + \log Z_\alpha$, where the strictness of the inequality follows from the fact that $\mathcal{X}=X+\mathcal{U}$ (which has a staircase density) cannot have density $\phi$.
Since $\alpha<1$ the latter inequality reads
\begin{equation}\label{gen-arikan}
G(X) > \frac{2^{\frac{1-\alpha}{\alpha}H_\alpha(X)} }{Z_\alpha^\frac{1-\alpha}{\alpha}}
\end{equation}
where $Z_\alpha=\int_{1/2}^{M+1/2} x^{-\frac{\alpha}{1-\alpha}}\d x$. 
In particular we have the following
\begin{corollary}[Arikan's Inequality~\cite{Arikan96}, slightly improved]
For $\alpha=\frac{1}{2}$, 
\begin{align}\label{ineq-arikan+}
G(X) &>  \frac{2^{H_{\frac{1}{2}}\!(X)}}{\ln(2M+1)}\\
G(X|Y) &>  \frac{2^{H_{\frac{1}{2}}\!(X|Y)}}{\ln(2M+1)}.  \label{ineq-arikan+|y}
\end{align} 
\end{corollary}
\begin{proof}
Plugging $Z_{\frac{1}{2}}\!=\!\ln \frac{M+1/2}{1/2}$ %
 in~\eqref{gen-arikan} gives~\eqref{ineq-arikan+}.
Since by~\eqref{def-arimoto}, $2^{H_{\frac{1}{2}}(X|Y)} = \E\, 2^{H_{\frac{1}{2}}(X|Y=y)}$, this immediately gives~\eqref{ineq-arikan+|y}.
\end{proof}

\begin{remark}
Inequality~\eqref{ineq-arikan+|y} slightly improves Arikan's original inequality~\eqref{ineq-arikan} for $M>1$ because $\ln(2M+1)<\ln(eM)=\ln M +1$.
It can be found from Arikan's derivation~\cite{Arikan96} that the optimal constant in the denominator is in fact $1+\frac{1}{2}+\frac{1}{3}+\cdots+\frac{1}{M}=\ln M + 0.5772\ldots + O(\frac{1}{M})$ (see~\cite[Eqn.~(47)]{SasonVerdu18bis}). Here $\ln(2M+1)=\ln M +  0.6931\ldots + O(\frac{1}{M})$ is not optimal but fairly close.
\end{remark}

For even smaller Rényi orders we have the following
\begin{corollary}
For any $0<\alpha<\frac{1}{2}$,
\begin{align}\label{gen-arikan<1/2}
G(X) &> (1-\tfrac{\alpha}{1-\alpha})^\frac{1-\alpha}{\alpha}\cdot \dfrac{ 2^{\frac{1-\alpha}{\alpha}H_\alpha(X)+\frac{1-2\alpha}{\alpha}}  }   {  {(2M+1)^\frac{1-2\alpha}{\alpha}}}
\\\label{gen-arikan<1/2|}
G(X|Y) &> (1-\tfrac{\alpha}{1-\alpha})^\frac{1-\alpha}{\alpha}\cdot \dfrac{ 2^{\frac{1-\alpha}{\alpha}H_\alpha(X|Y)+\frac{1-2\alpha}{\alpha}}  }   {  {(2M+1)^\frac{1-2\alpha}{\alpha}}}.
\end{align}
\end{corollary} 
\begin{proof}
One has $Z_\alpha= \int_{1/2}^{M+1/2} x^{-\frac{\alpha}{1-\alpha}}\d x < (M+\frac{1}{2})^{\frac{1-2\alpha}{1-\alpha}}$. Plugging this in~\eqref{gen-arikan} gives~\eqref{gen-arikan<1/2}. 
The second inequality then follows from~\eqref{def-arimoto}, which reads 
$2^{\frac{1-\alpha}{\alpha}H_\alpha(X|Y)} = \E\, 2^{\frac{1-\alpha}{\alpha}H_\alpha(X|Y=y)}$.
\end{proof}
 
\begin{example}
For any $M$-ary random variable $X$,
\begin{equation}
G(X) > \frac{2^{2H_{\frac{1}{3}}(X)} }{2(2M+1)}
\end{equation}
\begin{equation}
G(X) > \frac{32}{27}\cdot \frac{2^{3H_{\frac{1}{4}}(X)} }{(2M+1)^2}.
\end{equation}
and similarly for $X|Y$.
\end{example}

\begin{remark}
The method  of this Subsection also works for $\frac{1}{2}<\alpha<1$. In this case $Z_\alpha$ is bounded by $\frac{1-\alpha}{2\alpha-1} 2^{\frac{2\alpha-1}{1-\alpha}}$ (independently of $M$) and applying~\eqref{gen-arikan} gives
\begin{align}
G(X) &> (\tfrac{\alpha}{1-\alpha}\!-\!1)^\frac{1-\alpha}{\alpha}\cdot  2^{\frac{1-\alpha}{\alpha}H_\alpha(X)-\frac{2\alpha-1}{\alpha}}   \\[1ex]
G(X|Y) &> (\tfrac{\alpha}{1-\alpha}\!-\!1)^\frac{1-\alpha}{\alpha}\cdot  2^{\frac{1-\alpha}{\alpha}H_\alpha(X|Y)-\frac{2\alpha-1}{\alpha}}.
\end{align}
However, it can be verified that these inequalities are always weaker than~\eqref{ineq-rioul2-alpha} and~\eqref{ineq-rioul2-alpha|y}, respectively. This is not surprising since the derivation of the latter in the preceding subsection used, instead of~\eqref{suboptimalphi}, the optimal $\alpha$-exponential density achieving equality in~\eqref{ineq-general-continuous-alpha}.
\end{remark}

\subsection{Generalization to Guessing Moments}

While entropy $H(X)$ is generalized by the $\alpha$-entropy $H_\alpha(X)$ for any $\alpha>0$, the guessing entropy $G(X)$ can be generalized by the $\rho$-guessing entropy for any $\rho>0$, defined as the $\rho$th order moment~\cite{Arikan96}
\begin{equation}
G_\rho(X)\triangleq  \min \E\bigl(\truemathcal{G}^\rho(X)\bigr).
\end{equation}
Again the minimum occurs when the guessing function is a ranking function: $\truemathcal{G}(x)=k$ iff $p(x)=\P(X=x)$ is the $k$th largest probability in $X$'s distribution.
The conditional version given side information $Y$ is given by~\cite{Arikan96}
\begin{equation}
G_\rho(X|Y) \triangleq \E\bigl({G}_\rho(X|Y=y)\bigr).
\end{equation}

\begin{theorem}\label{thm-GrhoH}
When $H(X)$ is expressed in bits,
\begin{align}\label{GrhoH}
 G_\rho(X)  &> \frac{2^{\rho H(X)}}{\rho  \bigl(\Gamma(1+\tfrac{1}{\rho})\bigr)^\rho e}\\[1ex]
 G_\rho(X|Y)  &> \frac{2^{\rho H(X|Y)}}{\rho  \bigl(\Gamma(1+\tfrac{1}{\rho})\bigr)^\rho e}.
 \label{Grho|H}
\end{align} 
\end{theorem}

\begin{proof}
Applying Theorem~\ref{thm-one} to $\truemathcal{G}(X)$ for uniformly distributed $\mathcal{U}$ over the interval $(-1,0)$, one has $0\leq \mathcal{X}=\truemathcal{G}(X)+\mathcal{U}\leq \truemathcal{G}(X)$ with $h(\mathcal{X})=H(\truemathcal{G}(X))=H(X)$. Since $\theta=\E(\mathcal{X}^\rho)\leq \E(\truemathcal{G}^\rho(X))=G_\rho(X)$,~\eqref{ineq-gen-rho} of Theorem~\ref{lem-three} (one-sided case)
 gives~\eqref{GrhoH}.
The inequality is strict because the staircase density of $\truemathcal{G}(X)+\mathcal{U}$ cannot coincide with the (one-sided) $\alpha$-Gaussian achieving equality in~\eqref{ineq-gen-rho}.
Applying~\eqref{GrhoH} to $X|Y=y$ for every $y$, take the expectation over $Y$'s distribution and applying Jensen's inequality to the exponential function gives~\eqref{Grho|H}.
\end{proof}

\begin{remark}
During the revision process of this paper, the author became aware that~\eqref{GrhoH} (with an additional $o(1)$ term as $G_\rho(X)\to+\infty$) was obtained by Weinberger and Shayevitz in~\cite[Lemma~2]{WeinbergerShayevitz20} using a similar method.
\end{remark}

\begin{remark}
For $\rho=1$ we recover~\eqref{ineq-rioul2} without the additive constant $1/2$. 
This suboptimality comes from the fact that 
$\theta=\E(\mathcal{X}^\rho)=\E\bigl((\truemathcal{G}(X)+\mathcal{U})^\rho\bigr)$ cannot be determined as a function of $\E(\truemathcal{G}^\rho(X))=G_\rho(X)$ alone when $\rho\ne 1$.
\end{remark}

\begin{example}
For any discrete random variable $X$,
\begin{align}
G_2(X) &> 2 \cdot \frac{2^{2H(X)}}{\pi e}\\[1ex]
G_4(X) &> 8 \cdot \frac{2^{4H(X)}}{G^2\pi^3 e}
\end{align}
and similarly for $X|Y$, 
where $G%
=0.834626841674\ldots$ is Gauss's constant.
\end{example}

For $\alpha$-entropies we have the following
\begin{theorem}\label{thm-quinze}
When $H_\alpha(X)$ and $H_\alpha(X|Y)$ are expressed in bits, and $\alpha>\frac{1}{1+\rho}$, 
\begin{align}\label{GrhoHalpha}
G_\rho(X)&> 
\begin{cases}
\dfrac{\displaystyle 2^{\rho H_\alpha({X})}}{
  \bigl(\!\frac{(1+\rho)\alpha-1}{1-\alpha}\!\bigr)
 \bigl(\! \frac{\rho\alpha}{(1+\rho)\alpha-1} \!\bigr)^{\!\frac{\rho}{1-\alpha}}
 \Bigl(\! \frac{\Gamma(\frac{1}{\rho}+1)\Gamma(\frac{1}{1-\alpha}-\frac{1}{\rho})}{\Gamma(\frac{1}{1-\alpha})}\!\Bigr)^{\!\rho}}
\\\qquad\qquad\qquad\qquad\qquad\quad\text{for $\frac{1}{1+\rho}<\alpha<1$;}
\\[2ex] 
\dfrac{\displaystyle 2^{\rho H_\alpha({X})}}{
 \bigl(\!\frac{(1+\rho)\alpha-1}{\alpha-1}\!\bigr)
 \bigl(\frac{(1+\rho)\alpha-1}{\rho\alpha}\bigr)^{\frac{\rho}{\alpha-1}}
 \Bigl( \frac{\Gamma(\frac{1}{\rho}+1)\Gamma(\frac{\alpha}{\alpha-1})}{\Gamma(\frac{\alpha}{\alpha-1}+\frac{1}{\rho})}\Bigr)^{\!\rho}}
\\\qquad\qquad\qquad\qquad\qquad\quad\text{for $\alpha>1$,}
\end{cases}
\\\label{Grho|Halpha}
{G}_{\!\rho}(\!X|Y\!)\!&>\!\!
\begin{cases}
\!\dfrac{\displaystyle 2^{\rho H_\alpha({X|Y})}}{
\!  \bigl(\!\frac{(1+\rho)\alpha-1}{1-\alpha}\!\bigr)
 \bigl(\! \frac{\rho\alpha}{(1+\rho)\alpha-1} \!\bigr)^{\!\frac{\rho}{1-\alpha}}
\! \Bigl(\! \frac{\Gamma(\frac{1}{\rho}+1)\Gamma(\frac{1}{1-\alpha}-\frac{1}{\rho})}{\Gamma(\frac{1}{1-\alpha})}\!\Bigr)^{\!\rho}}
\\\qquad\qquad\qquad\qquad\qquad\quad\text{for $\frac{1}{1+\rho}<\alpha<1$;}
\\[2ex] 
\!\dfrac{\displaystyle 2^{\rho H_\alpha({X|Y})}}{
\! \bigl(\!\frac{(1+\rho)\alpha-1}{\alpha-1}\!\bigr)
 \bigl(\frac{(1+\rho)\alpha-1}{\rho\alpha}\bigr)^{\!\frac{\rho}{\alpha-1}}
 \Bigl( \frac{\Gamma(\frac{1}{\rho}+1)\Gamma(\frac{\alpha}{\alpha-1})}{\Gamma(\frac{\alpha}{\alpha-1}+\frac{1}{\rho})}\Bigr)^{\!\rho}}
\\\qquad\qquad\qquad\qquad\qquad\quad\text{for $\alpha>1$.}
\end{cases}
\end{align}
\end{theorem}

\begin{proof}
The proof of~\eqref{GrhoHalpha} is similar to the proof of~\eqref{GrhoH} in Theorem~\ref{thm-GrhoH} using inequality~\ref{eq-maxentgen-alpha} of
Theorem~\ref{lem-three} (one-sided case).
For~\eqref{Grho|Halpha} one proceeds as in the proof of Theorem~\ref{cor-ineq-rioul2-alpha} using~\eqref{def-arimoto} and the fact that $x\mapsto x^{\frac{\rho\alpha}{1-\alpha}}$ is strictly convex for all $\alpha>\frac{1}{1+\rho}$.
\end{proof}

\begin{example}
For any discrete random variable $X$,
\begin{align}
G_2(X) &>  \frac{2^{2H_{1/2}(X)}}{\pi^2}\\[1ex]
G_2(X) &> \frac{27}{16} \cdot \frac{2^{2H_{2/3}(X)}}{\pi^2}\\[1ex]
G_2(X) &> \frac{36}{125} \cdot {2^{2H_{2}(X)}}\\[1ex]
G_2(X) &> 3 \cdot \frac{2^{2H_{3}(X)}}{\pi^2}\\[1ex]
G_3(X) &> \frac{9}{2} \cdot \frac{2^{3H_{1/2}(X)}}{\sqrt{3}\,\pi^3}\\[1ex]
G_3(X) &> \frac{512}{2401} \cdot {2^{3H_{2}(X)}}\\[1ex]
G_4(X) &>  \frac{2^{4H_{1/3}(X)}}{G^4\pi^4}\\[1ex]
G_4(X) &>  \frac{27}{4} \cdot\frac{2^{4H_{1/2}(X)}}{\pi^4}\\[1ex]
G_4(X) &>  \frac{823543}{82944} \cdot\frac{2^{4H_{2/3}(X)}}{\pi^4}\\[1ex]
G_4(X) &>  \frac{10000}{59049} \cdot{2^{4H_{2}(X)}}\\[1ex]
G_4(X) &>  \frac{80}{9} \cdot\frac{2^{4H_{5}(X)}}{G^4\pi^4}
\end{align}
and similarly for $X|Y$, 
where $G%
=0.834626841674\ldots$ is Gauss's constant.
\end{example}

\begin{remark}
The reason why simple closed-form lower bounds on guessing entropy are obtained is due to the fact that Massey's approach uses bounds on continuous $\alpha$-entropies. Such simple lower bounds could not obtained by previous methods~\cite[Rmk.~5]{SasonVerdu18bis}. 
\end{remark}

\begin{remark}
While Theorem~\ref{thm-quinze} shows that $G_\rho(X)$ can always be lower-bounded by an exponential function of $H_{\alpha}(X)$ for any $\alpha>\frac{1}{1+\rho}$, such an inequality is impossible for $\alpha\leq\frac{1}{1+\rho}$ in general (when the number of possible values of $X$ is infinite). 
In fact, when $X$ has distribution $\P(X=k)=\frac{c}{(k\log k)^{\rho+1}}$\smallskip{} and $\alpha\leq \frac{1}{1+\rho}$, the series $\sum \frac{1}{k(\log k)^{\rho+1}}$ \smallskip{}converges---hence $G_\rho(X)$ is finite---while the series $\sum \frac{1}{(k\log k)^{\alpha(\rho+1)}}$ diverges so that $H_{\alpha}(X)=+\infty$.

As already remarked in~\cite[p.\,476]{LutwakYangShang05}, Arikan's inequality~\cite{Arikan96} on $G_\rho(X)$:
\begin{equation}
G_\rho(X) \geq  \frac{2^{H_{\!\frac{1}{1+\rho}}\!(X)}}{1\!+\!\ln M},
\end{equation}
(and similarly for $X|Y$) 
is for the limiting case $\alpha=\frac{1}{1+\rho}$, but  is valid only when $X$ takes a finite number $M$ of possible values. 
In a manner similar to was done in~\cite{SasonVerdu18bis},
it is always possible to use the method of Subsection~\ref{atifresm} to obtain inequalities of this kind for any $\alpha\leq\frac{1}{1+\rho}$.
\end{remark}

\section{Conclusion}\label{sec-conclusion}

Simple bounds on the differential entropy or Rényi entropy for a given fixed parameter (such as mean or variance) have long been established in connection with the important maximum entropy problem, which has been heavily studied for continuous distributions. By contrast, the similar problem for discrete distributions does not seem to be as popular: With the exception of discrete uniform or geometric laws, few results are known on the maximizing distributions.
However, bounding the discrete entropy or discrete Rényi entropy for a given fixed parameter (such as mean or variance) appears as a basic question in information theory. This paper has shown that using Massey's approach, many simple, closed-form bounds on discrete entropies or Rényi entropies can be deduced from bounds on the $\alpha$-entropies of a continuous distribution. One can envision that many similar derivations can be done for other types of parameter constraints.

Massey's approach gives, in particular, simple lower bounds on the guessing entropy or guessing moments, which are exponential in Rényi (or Rényi-Arimoto) entropies of any order $\alpha>0$, not just $\alpha=1$ (Massey's inequality) of $\alpha=\frac{1}{1+\rho}$ (Arikan's inequality). Since similar upper bounds also exist for $\alpha=\frac{1}{1+\rho}$~\cite{Arikan96,Bostas97,SasonVerdu18bis} it would be interesting to similarly upper bound guessing for other values of $\alpha$ in order to obtain tight evaluations in practical applications where a divide-and-conquer strategy is used~\cite{ChoudaryPopescu17} to guess a large secret from many small ones.

Finally, a variant of Massey's approach together with some Fourier analysis proves very tight ``Gaussian'' bounds for large variance---better than what would have been expected from convergence in entropy towards the Gaussian as established by the strong central limit theorem. Therefore, it is likely that Takano's $\sigma^{-1-\eps}$ term~\cite{Takano87} can be very much improved in general, at least for integer-valued random variables with finite higher-order moments. Since Massey-type bounds easily generalize to Rényi entropies with tight $\alpha$-Gaussian bounds, it would also be interesting to prove some corresponding convergence results in terms of $\alpha$-entropies and $\alpha$-Gaussians.

\section*{Acknowledgment}
The author is indebted to the anonymous reviewers for improving the clarity of exposition of this paper and pointing out references~\cite{SasonVerdu18bis} and~\cite{WeinbergerShayevitz20}.
\appendices %

\section{Reza's Equivalence Extended to Rényi Entropies}\label{app-reza}

Consider a continuous variable $\mathcal{X}$ having density~$f$, and {quantize} it to obtain the discrete~$X$ with step size~$\Delta$, in such a way that 
\begin{equation}
p(x_k)=\P(X\!=\!x_k)=\int_{k\Delta}^{(\!k+1\!)\Delta}\!f(x)\dx
\end{equation}
and the discrete values $x_k$ correspond to {mean} values  
\begin{equation}\label{mvt}
f(x_k)=\frac{1}{\Delta}\int_{k\Delta}^{(\!k+1\!)\Delta}\! \!f(x)\dx = \frac{p(x_k)}{\Delta}.
\end{equation}
\begin{proposition}\label{prop-reza}
If $f$ is continuous within each bin of length~$\Delta$ and the integral (in~\eqref{eq-h} or in~\eqref{eq-h-alpha}) defining $h_\alpha(X)$ exists, then
$$
\lim\limits_{\Delta\to 0} \{ H_\alpha(X)+\log\Delta \}= h_\alpha(\mathcal{X}).
$$
\end{proposition}
The assumptions are satisfied in particular when $f$ is continuous and compactly supported.

\begin{proof}
By the continuity assumption, the values~\eqref{mvt} are well defined and given by the mean value theorem.
Since the integral in~\eqref{eq-h} (resp.~\eqref{eq-h-alpha}) converges and $f$ is piecewise continuous, $f \log f$ (resp. $f^\alpha$) is Riemann-integrable. 
It follows that the integral in~\eqref{eq-h} and in~\eqref{eq-h-alpha} can be respectively approximated by the Riemann sum
\begin{align}
\sum_k \Delta\!\cdot\! f(x_k) \log \frac{1}{f(x_k)}  &= \sum_k p(x_k) \log \frac{\Delta}{p(x_k)}\notag\\&=H(X)+\log \Delta \\[1ex]
\tfrac{1}{1-\alpha} \log\sum_k \Delta\!\cdot\! f^\alpha(x_k)  &=\tfrac{1}{1-\alpha} \log  \sum_k \Delta^{1-\alpha} p^\alpha(x_k) \notag\\&=H_\alpha(X)+\log \Delta,
\end{align}
which tends to $h(\mathcal{X})$ (resp. $h_\alpha(\mathcal{X})$)  as $\Delta\to 0$.
\end{proof}

\section{Massey's Equivalence Extended to Rényi Entropies and Arbitrary Step Size}
\label{app-mass}

\begin{proof}[Proof of Theorem~\ref{thm-one}]
The density of $\mathcal{X}=X+ \mathcal{U}$ is a mixture of the form
\begin{equation}
f(x) = \sum_{k\in\Z} p(x_k) \,\chi(x-x_k) 
\end{equation}
where $x_k$ are the regularly spaced values of $X$ and $\chi$ is the density of $\mathcal{U}$.
The terms in the sum have disjoint supports.
Since entropy is invariant by translation, we may always assume that $\chi$ is supported in the interval $[0,\Delta]$. Splitting the integral in~\eqref{eq-h} or in~\eqref{eq-h-alpha} into parts over intervals $[x_k,x_{k+1}=x_k+\Delta]$ we obtain
\begin{equation}
\begin{aligned}
h(\mathcal{X})&=\sum_k p(x_k)\int \!\!\chi(x-x_k) \log\! \frac{1}{p(x_k)\chi(x-x_k)} \dx \\&= 
\sum_k p(x_k)\Bigl[\underbrace{\int\!\!\chi}_{=1}\Bigr] \log\! \frac{1}{p(x_k)}+ \Bigl[\underbrace{\sum_k p(x_k)}_{=1}\Bigr]\int \!\!\chi \log \frac{1}{\chi}
\\
h_\alpha(\mathcal{X}) &= \frac{1}{1-\alpha} \log \sum_k p(x_k)^\alpha\int\!\! \chi(x-x_k)^\alpha \dx \\&=
\frac{1}{1-\alpha} \log \sum_k p(x_k)^\alpha 
\int\!\! \chi^\alpha
\\&=
\frac{1}{1-\alpha} \log \sum_k p(x_k)^\alpha + \frac{1}{1-\alpha} \log 
\int\! \chi^\alpha
\end{aligned}
\end{equation}
which proves~\eqref{eq-X+Z}.
\end{proof}

\begin{remark}
The above proof follows the textbook solution~\cite{CoverThomasSolutions} to  exercice 8.7 of~\cite{CoverThomas} in the case $\alpha=1$. (A similar calculation appears in~\cite[Proof of Thm.~3]{RioulMagossi14}.)
In this particular case, an even simpler proof is as follows.
\begin{proof}[Proof of Theorem 1 ($\alpha=1$)]
By the support assumption, $X$ can be recovered by rounding $X+\mathcal{U}$, hence is a deterministic function of $\mathcal{X}$. Therefore,
$H(X|\mathcal{X})=0$ and
\begin{equation}
\begin{aligned}
H(X)&=H(X)-H(X|\mathcal{X})
\\&=I(X;\mathcal{X})\\&=h(\mathcal{X})-h(\mathcal{X}|X)\\&=h(\mathcal{X})-h(\mathcal{U}),  
\end{aligned}
\end{equation}
 which proves~\eqref{eq-X+Z}.
\end{proof}
\end{remark}

\section{Proof of Theorem~\ref{lem-three} and Its Corollaries}\label{app-lem}

We first prove Theorem~\ref{lem-three} and then deduce Corollaries~\ref{lem-one} and \ref{lem-two} as particular cases.

Set $T(x)$ in the form $T(x)=1+\beta\frac{|x|^\rho}{\theta}$ so that $m=1+\beta$ and $\beta$ is such that~\eqref{eq-phi-alpha} has finite $\rho$th-order moment $\theta=\E(|\mathcal{X}|^\rho)$. 
In order that $\phi(x)=\frac{1}{Z}\bigl(1+\beta\frac{|x|^\rho}{\theta}\bigr)^{\frac{1}{\alpha-1}}$ be integrable, it is necessary that $\beta$ has the same sign as $1-\alpha$. 

For $\alpha>1$ ($\beta<0$), the density is supported in the interval $|x|<\sqrt[\leftroot{-2}\uproot{2}\rho]{\frac{\smash[t]{\theta}}{\smash[b]{|\beta|}}}$ so that $1+\beta\frac{|x|^\rho}{\theta}\geq 0$.
In this case we write $\phi(x)=\frac{1}{Z} \big(1-|\beta|\frac{|x|^\rho}{\theta}\bigr)_{\!+}^{\frac{1}{\alpha-1}}$
with the notation $(X)_+=\max(X,0)$.

For $\alpha<1$, the existence of a finite variance implies that the integral of $\bigl(1+\beta\frac{|x|^\rho}{\theta}\bigr)^{1-\frac{1}{1-\alpha}}$ converges at infinity, which requires $\alpha>\frac{1}{1+\rho}$. 

In either case, $\beta$ is such that $\phi$ has $\rho$th moment $\theta$, that is, such that~\eqref{eq-trick} holds, hence $m=1+\beta=\frac{Z_\alpha}{Z}$. Now we can write
$Z= \sqrt[\leftroot{-2}\uproot{2}\rho]{\frac{\smash[t]{\theta}}{\smash[b]{|\beta|}}}\, I\bigl(\frac{1}{\alpha-1}\bigr)$ and
$Z_\alpha= \sqrt[\leftroot{-2}\uproot{2}\rho]{\frac{\smash[t]{\theta}}{\smash[b]{|\beta|}}}\, I\bigl(\frac{\alpha}{\alpha-1}\bigr)$
where %
\begin{equation}
I(\gamma)\triangleq
\begin{cases}
\displaystyle\int_{-\infty}^{+\infty}\!  \frac{\mathrm{d}x}{(1+|x|^\rho)^{-\gamma}}=\frac{2}{\rho}\int_0^1\!\! (1-t)^{-\gamma-\frac{3}{2}}t^{-\frac{1}{2}} \dt 
\\[2ex]=\dfrac{\frac{2}{\rho}\Gamma(\frac{1}{\rho})\Gamma(-\gamma-\frac{1}{\rho})  }{ \Gamma(-\gamma)}
\qquad\qquad\text{for $\gamma<0$;}\\[3ex]
\displaystyle\int_{-1}^1 (1-|x|^\rho)^\gamma \dx \,=\,\frac{2}{\rho}\int_0^1 (1-t)^\gamma t^{-\frac{1}{2}} \dt   
\\[2ex]=\dfrac{\frac{2}{\rho}\Gamma(\frac{1}{\rho})\Gamma(\gamma+1) }{ \Gamma(\gamma+1+\frac{1}{\rho})}
\quad\qquad\qquad\text{for $\gamma>0$.}
\end{cases}
\end{equation}
Here we have made the change of variables $t=\frac{x^p}{1+x^p}$ and $t=x^p$, respectively, for $x>0$, and recognized Euler integrals of the first kind. In either case, letting $\gamma=\frac{1}{\alpha-1}$,
\begin{equation}\label{eq-var-m}
m=\frac{Z_\alpha}{Z}=\frac{I(\gamma+1)}{I(\gamma) }
=\frac{-\gamma-1}{-\gamma-1-\frac{1}{\rho}}
=\frac{\rho\alpha}{(\rho+1)\alpha-1},
\end{equation}
hence $\beta=\frac{1-\alpha}{(\rho+1)\alpha-1}$.
Plugging this and the expression of $Z$ into that of $\phi$ gives the expression of the generalized $\alpha$-Gaussian density~\cite{LutwakYangShang05}\footnote{%
There is a misprint in the expression\smallskip{} of the generalized $\alpha$-Gaussian in~\cite[p.\,474]{LutwakYangShang05} where $\beta(\frac{1}{p},\frac{1}{1-\lambda})$ should read $\beta(\frac{1}{p},\frac{\lambda}{\lambda-1})$.
}:
\begin{equation}\label{alpha-Gaussian-generalized}
\phi(x) = 
\begin{cases}
 \sqrt[\leftroot{-2}\uproot{2}\rho]{\dfrac{{\beta}}{{\theta}}} \;
\dfrac{\Gamma(\frac{1}{1-\alpha})}{2\Gamma(1+\frac{1}{\rho})\Gamma(\frac{1}{1-\alpha}-\frac{1}{\rho})} \;
\dfrac{1}{\big(  1+\beta\frac{|x|^\rho}{\theta}    \bigr)^{\frac{1}{1-\alpha}}}
\\\qquad\qquad\qquad\qquad\qquad\qquad\text{for $\frac{1}{1+\rho}<\alpha<1$;}\\[3ex]
 \sqrt[\leftroot{-2}\uproot{2}\rho]{\dfrac{{|\beta|}}{{\theta}}} \;
\dfrac{\Gamma(\frac{\alpha}{\alpha-1}+\frac{1}{\rho})}{2\Gamma(1+\frac{1}{\rho})\Gamma(\frac{\alpha}{\alpha-1})} \;
  \big(1-|\beta|\frac{|x|^\rho}{\theta}\bigr)_{\!+}^{\frac{1}{\alpha-1}}
\\\qquad\qquad\qquad\qquad\qquad\qquad\text{for $\alpha>1$,}
\end{cases}
\end{equation}
and plugging~\eqref{eq-var-m} and the expression of $Z_\alpha$ or $Z$ into~\eqref{ineq-general-continuous-alpha} or~\eqref{ineq-general-continuous-alpha-Z} gives~\eqref{eq-maxentgen-alpha}. 
\qed

Corollary~\ref{lem-one} follows by setting $\rho=2$ for the centered variable $\mathcal{X}-\mu_{\mathcal{X}}$. 
The corresponding expression of the $\alpha$-Gaussian density (with $\beta=\frac{1-\alpha}{3\alpha-1}$) is~\cite{CostaHeroVignat03}
\begin{equation}\label{alpha-Gaussian}
\phi(x) = 
\begin{cases}
\sqrt{\dfrac{\beta}{\pi\sigma^2_{\mathcal{X}}}} \;
\dfrac{\Gamma(\frac{1}{1-\alpha})}{\Gamma(\frac{1}{1-\alpha}-\frac{1}{2})} \;
\dfrac{1}{\big(1+\beta (\frac{x-\mu_{\mathcal{X}}}{\sigma_{\mathcal{X}}})^2\bigr)^{\frac{1}{1-\alpha}}}
\\\qquad\qquad\qquad\qquad\qquad\qquad
\text{for $\frac{1}{3}<\alpha<1$;}\\[3ex]
\sqrt{\dfrac{|\beta|}{\pi\sigma^2_{\mathcal{X}}}}\;
\dfrac{\Gamma(\frac{\alpha}{\alpha-1}+\frac{1}{2})}{\Gamma(\frac{\alpha}{\alpha-1})} \;
  \big(1-|\beta| (\frac{x-\mu_{\mathcal{X}}}{\sigma_{\mathcal{X}}})^2\bigr)_{\!+}^{\frac{1}{\alpha-1}}
\\\qquad\qquad\qquad\qquad\qquad\qquad
\text{for $\alpha>1$,}
\end{cases}
\end{equation}
and $Z_\alpha$ is given by
\begin{equation}\label{eq-Zalpha-sigma}
Z_\alpha= \tfrac{\sigma_{\mathcal{X}}}{\sqrt{|\beta|}} I\bigl(\tfrac{\alpha}{\alpha-1}\bigr)=
\begin{cases}
\sqrt{\frac{\pi\sigma^2_{\mathcal{X}}(3\alpha-1)}{{1-\alpha}}}
\frac{\Gamma(\frac{\alpha}{1-\alpha}-\frac{1}{2})  }{ \Gamma(\frac{\alpha}{1-\alpha})}
&\text{for $\alpha<1$;}\\[1ex]
\sqrt{\frac{\pi\sigma^2_{\mathcal{X}}(3\alpha-1)}{{\alpha-1}}}
\frac{\Gamma(\frac{\alpha}{\alpha-1}+1) }{ \Gamma(\frac{\alpha}{\alpha-1}+\frac{3}{2})}
&\text{for $\alpha>1$.}
\end{cases}
\end{equation}

Corollary~\ref{lem-two} follows by setting $\rho=1$, where the multiplying factor $2$ in the above expressions is removed due to the one-sided constraint $\mathcal{X}\geq 0$.
The corresponding expression of the ``$\alpha$-exponential'' density (with $\beta=\frac{1-\alpha}{2\alpha-1}$) for $x>0$ is~\cite[\S\,II.B]{BunteLapidoth14}
\begin{equation}\label{alpha-exp}
\phi(x)=
\begin{cases}
\dfrac{\beta}{\mu_{\mathcal{X}}}\dfrac{\alpha}{1-\alpha}
\dfrac{1}{\big(1+\beta \frac{x}{\mu_{\mathcal{X}}}\bigr)^{\frac{1}{1-\alpha}}}
&\text{for $\frac{1}{2}<\alpha<1$}\\
\dfrac{|\beta|}{\mu_{\mathcal{X}}}\dfrac{\alpha}{\alpha-1}
 \big(1-|\beta| \frac{x}{\mu_{\mathcal{X}}}\bigr)_{\!+}^{\frac{1}{\alpha-1}}
&\text{for $\alpha>1$.}
\end{cases}
\end{equation} 
and $Z_\alpha$ is given by
\begin{equation}\label{eq-Zalpha-mean}
Z_\alpha= \frac{\mu_{\mathcal{X}}}{{|\beta|}} I\bigl(\frac{\alpha}{\alpha-1}\bigr)=\frac{\mu_{\mathcal{X}}}{{\beta}} \frac{1-\alpha}{2\alpha-1}=  \mu_{\mathcal{X}}.
\end{equation}

\section{Proof of Inequality~(\ref{trapezoidal})}
\label{app-trapezoidal}

Let $s=\frac{\alpha}{|\alpha-1|}$ and $a=\frac{2\alpha-1}{|1-\alpha|}\mu$. Then~\eqref{trapezoidal} is equivalent to 
\begin{align}
\sum_{x\in\N}\Bigl(1+\frac{x}{a}\Bigr)^{-s} &> \frac{a}{s-1}+\frac{1}{2} \qquad (\frac{1}{2}<\alpha<1)
\\
\sum_{x\in\N}\Bigl(1-\frac{x}{a}\Bigr)_{\!\!+}^{s} &> \frac{a}{s+1}+\frac{1}{2} \qquad (\alpha>1) 
\end{align}
This is proved by applying the following Lemma to $f(x)=\bigl(1+\frac{x}{a}\bigr)^{-s}$ and $\bigl(1-\frac{x}{a}\bigr)_{\!\!+}^{s}$, respectively.

\begin{lemma}
Let $f$ be nonnegative decreasing and strictly convex. Then
\begin{equation}
\sum_{x\in\N} f(x)  >  \frac{f(0)}{2}+\int_0^{+\infty} f(x) \,\mathrm{d}x.
\end{equation}
\end{lemma}

\begin{proof}
Let $g(x)$ be the piecewise linear function defined for all $x\geq 0$ that linearly interpolates  the values of $f$ over the integers. Then
$\int_0^{+\infty} f(x) \,\mathrm{d}x <\int_0^{+\infty} g(x) \,\mathrm{d}x = \sum_{x\in\N} \frac{f(x)+f(x+1)}{2} = \sum_{x\in\N} f(x) - \frac{f(0)}{2}$.
\end{proof}

\IEEEtriggeratref{27}
%
%
%
%


%
%
%

\vspace*{-2ex}

\begin{IEEEbiographynophoto}{Olivier Rioul} is full Professor at the Department of Communication and Electronics at Télécom Paris, Institut Polytechnique de Paris, France. He graduated from École Polytechnique, Paris, France in 1987 and from École Nationale Supérieure des Télécommunications, Paris, France in 1989. He obtained his PhD degree from École Nationale Supérieure des Télécommunications, Paris, France in 1993. His research interests are in applied mathematics and include various, sometimes unconventional, applications of information theory such as inequalities in statistics, hardware security, and experimental psychology. He has been teaching information theory at various universities for almost twenty years and has published a textbook which has become a classical French reference in the field. 
\end{IEEEbiographynophoto}
\vfill

\end{document}